%% file: main_combined.tex
\RequirePackage[T1]{fontenc}
\documentclass[journal]{IEEEtran}
\usepackage{iftex}
\ifPDFTeX
  \usepackage[utf8]{inputenc}
  
\fi
\usepackage{amsmath}
\usepackage{balance}
\newcommand{\TightDisplaySpacing}{%
  \setlength{\textfloatsep}{2pt plus 1pt minus 1pt}%
  \setlength{\abovecaptionskip}{1pt}%
  \setlength{\belowcaptionskip}{1pt}%
  \setlength{\abovedisplayskip}{2pt plus 1pt minus 1pt}%
  \setlength{\belowdisplayskip}{2pt plus 1pt minus 1pt}%
  \setlength{\abovedisplayshortskip}{1pt plus 1pt minus 1pt}%
  \setlength{\belowdisplayshortskip}{1pt plus 1pt minus 1pt}%
  \setlength{\jot}{1pt}%
}
\AtBeginDocument{\TightDisplaySpacing}
\usepackage{array}
\usepackage[table]{xcolor}
\IfFileExists{circledsteps.sty}{\usepackage{circledsteps}}{}
\usepackage{amssymb}
\usepackage{amsthm}
\usepackage{graphicx}
\usepackage{cite}
\usepackage{bm}
\usepackage{comment}
\usepackage{epsfig,psfrag}
\usepackage{tabularx}
\usepackage{acronym}
\usepackage[yyyymmdd,hhmmss]{datetime}
\usepackage{bbm}
\usepackage{notation}
\usepackage{textgreek}
\usepackage{mathrsfs}
\usepackage{multirow}
\usepackage{colortbl,pgfplotstable}
\usepackage{booktabs}
\usepackage{pgfplots}
\usepackage{tikz}
\usetikzlibrary{arrows,arrows.meta,shapes.arrows,calc,backgrounds,intersections}

\usepackage[caption=false,font=footnotesize]{subfig}
\usepackage{placeins}

\usepackage{algpseudocode}
\let\algpseudoForAll\ForAll
\let\algpseudoEndFor\EndFor
\let\algpseudoIf\If
\let\algpseudoElse\Else
\let\algpseudoEndIf\EndIf
\let\algpseudoState\State
\let\algpseudoRequire\Require
\let\algpseudoEnsure\Ensure
\let\algpseudoComment\Comment
\usepackage{algorithm2e}
\newcommand{\usealgpseudocodecommands}{%
  \let\ForAll\algpseudoForAll
  \let\EndFor\algpseudoEndFor
  \let\If\algpseudoIf
  \let\Else\algpseudoElse
  \let\EndIf\algpseudoEndIf
  \let\State\algpseudoState
  \let\Require\algpseudoRequire
  \let\Ensure\algpseudoEnsure
  \let\Comment\algpseudoComment
}
\usepackage{xstring}
\usepackage{bibunits}
\usepackage{hyperref}

\makeatletter
\let\CombinedNewLabel\newlabel
\renewcommand{\newlabel}[2]{%
  \CombinedNewLabel{#1}{#2}%
  \IfBeginWith{#1}{acro:}{}{%
    \CombinedNewLabel{main-#1}{#2}%
    \CombinedNewLabel{supp-#1}{#2}%
  }%
}
\renewcommand{\@bibunitname}{\jobname-bib\the\@bibunitauxcnt}
\AtBeginDocument{%
  \renewcommand{\bibcite}[2]{%
    \global\@namedef{b@#1\@extra@b@citeb}{%
      \hyper@@link[cite]{}{cite.#1\@extra@b@citeb}{#2}%
    }%
  }%
}
\makeatother
\hypersetup{
    colorlinks=true,
    urlcolor=black,
    linkcolor=black,
    citecolor=black,
    filecolor=black
}

\allowdisplaybreaks
\newcommand{\linkedref}[2]{\hyperref[#2]{#1~\ref*{#2}}}
\newcommand{\linkedrefpair}[3]{\hyperref[#2]{#1~\ref*{#2}} and~\hyperref[#3]{\ref*{#3}}}
\newcommand{\linkedrefrange}[3]{\hyperref[#2]{#1~\ref*{#2}}--\hyperref[#3]{\ref*{#3}}}
\newcommand{\secref}[1]{\linkedref{Section}{#1}}

\newcommand{\figref}[1]{\linkedref{Fig.}{#1}}

\newcommand{\tabref}[1]{\linkedref{Table}{#1}}
\newcommand{\Tabref}[1]{\tabref{#1}}
\newtheorem{definition}{Definition}
\newtheorem{proposition}{Proposition}
\newtheorem{assumption}{Assumption}
\newtheorem{remark}{Remark}

\newcolumntype{L}[1]{>{\raggedright\arraybackslash}p{#1}}
\newcolumntype{C}[1]{>{\centering\arraybackslash}p{#1}}
\newcolumntype{R}[1]{>{\raggedleft\arraybackslash}p{#1}}
\definecolor{BLUE}{rgb}{0,0,1}

\definecolor{niceblue}{RGB}{0,154,255}
\definecolor{plotC}{HTML}{1974D2} 
\definecolor{plotCO}{HTML}{005A9C} 
\definecolor{plotCP}{HTML}{0B3C5D} 
\definecolor{plotCL}{HTML}{17BECF} 
\definecolor{plotD}{HTML}{FF7F0E} 
\definecolor{plotDO}{HTML}{B35C00} 
\definecolor{plotETT}{HTML}{7F7F7F} 
\definecolor{plotHM}{HTML}{2E8B57} 
\definecolor{plotHWM}{HTML}{CC0000} 
\definecolor{plotHP}{HTML}{8C564B} 
\definecolor{plotHL}{HTML}{9467BD} 
\colorlet{plotHMO}{plotHWM} 
\definecolor{plotHWTM}{HTML}{00B050} 

\newcommand{\DocumentLegendFont}{\fontsize{7.2}{7.8}\selectfont}

\newcommand{\R}{\mathbb{R}}

\providecommand{\bl}[1]{#1}
\renewcommand{\bl}[1]{#1}
\newcommand{\blsuppref}[2]{\hyperref[supp-#1]{#2}}
\newcommand{\bleqref}[1]{\hyperref[supp-#1]{(\ref*{supp-#1})}}

\providecommand{\rev}[1]{#1}
\newenvironment{revision}{}{}
\definecolor{forestgreen}{RGB}{34,139,34}

\DeclareRobustCommand{\revHyowon}[1]{{#1}}
\DeclareRobustCommand{\revYu}[1]{{#1}}
\DeclareRobustCommand{\revGrammar}[1]{{#1}}
\DeclareRobustCommand{\revGrBP}[1]{{#1}}
\DeclareRobustCommand{\grbp}{\revGrBP{\acs{GrBP}}}
\DeclareRobustCommand{\grbpv}[1]{\revGrBP{\mbox{\acs{GrBP}-#1}}}
\input{acronyms}

\pgfplotsset{compat=1.14}
\pgfplotsset{
  every axis legend/.append style={
    font=\DocumentLegendFont
  }
}
\pgfplotsset{
  start from first ho/.style={
    row predicate/.code={%
      \ifnum\pdfstrcmp{\thisrow{frame}}{summary}=0\relax
        \pgfplotstableuserowfalse
      \else
        \ifx\FirstNonZeroSeen\relax
          \pgfmathparse{\thisrow{ho_count}>0}
          \ifnum\pgfmathresult=1
            \pgfplotstableuserowtrue
            \def\FirstNonZeroSeen{1}
          \else
            \pgfplotstableuserowfalse
          \fi
        \else
          \pgfplotstableuserowtrue
        \fi
      \fi
    }
  }
}
\pgfplotsset{
  start from first active track/.style={
    row predicate/.code={%
      \ifnum\pdfstrcmp{\thisrow{frame}}{summary}=0\relax
        \pgfplotstableuserowfalse
      \else
        \ifx\FirstActiveTrackSeen\relax
          \pgfmathparse{\thisrow{track_count_avg}>0}%
          \ifnum\pgfmathresult=1
            \pgfplotstableuserowtrue
            \def\FirstActiveTrackSeen{1}%
          \else
            \pgfplotstableuserowfalse
          \fi
        \else
          \pgfplotstableuserowtrue
        \fi
      \fi
    }
  }
}
\pgfplotsset{
  start from first active monte carlo track/.style={
    row predicate/.code={%
      \ifnum\pdfstrcmp{\thisrow{frame}}{summary}=0\relax
        \pgfplotstableuserowfalse
      \else
        \pgfmathparse{\thisrow{track_count_contrib}>1}%
        \ifnum\pgfmathresult=1
          \ifx\FirstActiveTrackSeen\relax
            \pgfmathparse{\thisrow{track_count_avg}>0}%
            \ifnum\pgfmathresult=1
              \pgfplotstableuserowtrue
              \def\FirstActiveTrackSeen{1}%
            \else
              \pgfplotstableuserowfalse
            \fi
          \else
            \pgfplotstableuserowtrue
          \fi
        \else
          \pgfplotstableuserowfalse
        \fi
      \fi
    }
  }
}

\begin{document}

\begin{bibunit}[IEEEtran]
\makeatletter
\def\@extra@b@citeb{.main}
\@startbibunitorrelax
\makeatother
\bstctlcite[@bibunitaux]{IEEEexample:BSTcontrol}

\title{Scalable Extended-Target Handover in Distributed Integrated Sensing and Communication}

\author{Liping Bai,~\IEEEmembership{Student Member,~IEEE}, Yu Ge,~\IEEEmembership{Member,~IEEE}, Ali Al Khansa, \\Hyowon Kim,~\IEEEmembership{Member,~IEEE}, Erik Leitinger,~\IEEEmembership{Member,~IEEE}, and Henk Wymeersch,~\IEEEmembership{Fellow,~IEEE}%
\thanks{L. Bai and H. Wymeersch are with the Department of Electrical Engineering, Chalmers University of Technology, 412
96 Gothenburg, Sweden\revGrammar{.}}%
\thanks{Y. Ge is with the Research Laboratory of Electronics, Massachusetts Institute of Technology, 02139--4307 Massachusetts, United States of America.}
\thanks{A. Al Khansa is with Orange Labs, 35510 Rennes, France.}
\thanks{H. Kim is with Chungnam National University, Daejeon 34134, South Korea.}%
\thanks{E. Leitinger is with the Institute of
Communication Networks and Satellite Communications, Graz University of
Technology, Graz 8010, Austria.}
\thanks{This work was supported by the Wallenberg Foundations through the Wallenberg AI, Autonomous Systems and Software Program and by the SNS JU project 6G-DISAC under the EU's Horizon Europe research and innovation program under Grant Agreement No. 101139130.}%
\thanks{Part of this work has been presented in the conference paper~\cite{bai2025beliefpropagationbasedtargethandover}.}
\thanks{OpenAI Codex was used for editorial refinement of the manuscript and for suggesting refactorings and implementation optimizations to the code. All suggestions were critically reviewed, edited where appropriate, and validated by the authors. The authors retain full responsibility for the manuscript, software, experimental design, results, and conclusions.}}

\maketitle

\begin{abstract}
\Ac{DISAC} networks require
multi-target tracking methods whose per-node computation and inter-base-station
communication remain bounded as both the network and target populations grow.
Existing multisensor fusion methods provide \revGrammar{a} strong theoretical foundation, but
scalability is seldom treated as the primary design objective. \revGrBP{We develop a grouped-measurement belief-propagation \ac{MTT} method and combine it with a proposed event-triggered target-handover protocol.} When a tracked target is predicted to become observable at another
base station, the owner transfers a predicted track message while retaining its
local copy. Under bounded local workload and
neighbor degree, we show that all variants of the proposed handover methods have network-size-independent
per-node loads. Simulations demonstrate that handover reduces boundary track
loss relative to uncoordinated processing and approaches coordinated-processing
accuracy with lower communication overhead.
\end{abstract}


\input{sections/sec_introduction}
\input{sections/sec_system_model}
\input{sections/sec_detection_density}

\input{sections/sec_reference_methods}
\input{sections/sec_handover}

\input{sections/sec_numerical_results}
\input{sections/sec_conclusion}
\appendices
\begingroup
\setcounter{equation}{0}%
\renewcommand{\theequation}{\thesection.\arabic{equation}}%
\renewcommand{\theHequation}{appendix\thesection.\arabic{equation}}%
\input{sections/sec_appendix_proofs}
\endgroup

\clearpage
\end{bibunit}

\begin{bibunit}[IEEEtran]
\makeatletter
\def\@extra@b@citeb{.supplement}
\@startbibunitorrelax
\makeatother
\setcounter{section}{0}
\setcounter{subsection}{0}
\setcounter{subsubsection}{0}
\setcounter{equation}{0}
\setcounter{figure}{0}
\setcounter{footnote}{0}
\renewcommand{\theHsection}{supplement.\arabic{section}}
\renewcommand{\theHfootnote}{supplement.\arabic{footnote}}
\renewcommand{\thefigure}{S\arabic{figure}}
\renewcommand{\theHfigure}{supplement.\arabic{figure}}
\renewcommand{\theequation}{S\text{-}A.\arabic{equation}}
\renewcommand{\theHequation}{supplement.SA.\arabic{equation}}
\renewcommand{\thesubsection}{S-A.\Alph{subsection}}
\renewcommand{\theHsubsection}{supplement.SA.\Alph{subsection}}
\renewcommand{\theHsubsubsection}{supplement.SA.\Alph{subsection}.\arabic{subsubsection}}

\bstctlcite[@bibunitaux]{IEEEexample:BSTcontrol}
\title{\LARGE Supplementary Material\\[0.4em]
\Large Scalable Extended-Target Handover in Distributed\\
Integrated Sensing and Communication}
\author{Liping Bai,~\IEEEmembership{Student Member,~IEEE}, Yu Ge,~\IEEEmembership{Member,~IEEE}, Ali Al Khansa, \\
Hyowon Kim,~\IEEEmembership{Member,~IEEE}, Erik Leitinger,~\IEEEmembership{Member,~IEEE}, and Henk Wymeersch,~\IEEEmembership{Fellow,~IEEE}%
\thanks{L. Bai and H. Wymeersch are with the Department of Electrical Engineering, Chalmers University of Technology, 412
96 Gothenburg, Sweden.}%
\thanks{Y. Ge is with the Research Laboratory of Electronics, Massachusetts Institute of Technology, 02139--4307 Massachusetts, United States of America.}
\thanks{A. Al Khansa is with Orange Labs, 35510 Rennes, France.}
\thanks{H. Kim is with Chungnam National University, Daejeon 34134, South Korea.}%
\thanks{E. Leitinger is with the Institute of
Communication Networks and Satellite Communications, Graz University of
Technology, Graz 8010, Austria.}}
\maketitle

\noindent\textit{Section and equation references without an ``S'' prefix refer to the main article.}

\input{sections/sec_appendixI}


\clearpage
\end{bibunit}
\end{document}

%% file: acronyms.tex
\makeatletter
\newcommand{\AcronymList}{}
\newcommand{\DeclareAcronym}[3]{%
  \acrodef{#1}[#2]{#3}%
  \g@addto@macro\AcronymList{%
    \acro{#1}[#2]{#3}%
  }%
}
\makeatother

\DeclareAcronym{AoA}{AoA}{angle of arrival}
\DeclareAcronym{ADMM}{ADMM}{alternating direction method of multipliers}
\DeclareAcronym{BP}{BP}{belief propagation}
\DeclareAcronym{BS}{BS}{base station}
\DeclareAcronym{CPHD}{CPHD}{cardinalized probability hypothesis density}
\DeclareAcronym{CTRV}{CTRV}{constant turn rate and constant velocity}
\DeclareAcronym{DBSCAN}{DBSCAN}{Density-Based Spatial Clustering of Applications with Noise}
\DeclareAcronym{DISAC}{DISAC}{distributed integrated sensing and communication}
\DeclareAcronym{DMTT}{DMTT}{distributed multi-target tracking}
\DeclareAcronym{DSN}{DSN}{distributed sensor network}
\DeclareAcronym{EKF}{EKF}{extended Kalman filter}
\DeclareAcronym{ET}{ET}{\revGrammar{extended target}}
\DeclareAcronym{ETT}{ETT}{extended-target tracking}
\DeclareAcronym{FOV}{FoV}{field of view}
\DeclareAcronym{GCI}{GCI}{generalized covariance intersection}
\DeclareAcronym{GLMB}{GLMB}{generalized labeled multi-Bernoulli}
\DeclareAcronym{GGIW}{GGIW}{gamma Gaussian inverse Wishart}
\DeclareAcronym{GOSPA}{GOSPA}{generalized optimal subpattern assignment}
\DeclareAcronym{GPU}{GPU}{graphics processing unit}
\DeclareAcronym{ISAC}{ISAC}{integrated sensing and communication}
\DeclareAcronym{LMB}{LMB}{labeled multi-Bernoulli}
\DeclareAcronym{MIMO}{MIMO}{multiple-input multiple-output}
\DeclareAcronym{MOT}{MOT}{multi-object tracking}
\DeclareAcronym{OFDM}{OFDM}{orthogonal frequency-division multiplexing}
\DeclareAcronym{PMF}{PMF}{probability mass function}
\DeclareAcronym{MTT}{MTT}{multi-target tracking}
\DeclareAcronym{PDF}{PDF}{probability density function}
\DeclareAcronym{PHD}{PHD}{probability hypothesis density}
\DeclareAcronym{GrBP}{GrBP}{grouped-measurement belief propagation}
\DeclareAcronym{PMB}{PMB}{Poisson multi-Bernoulli}
\DeclareAcronym{PMBM}{PMBM}{Poisson multi-Bernoulli mixture}
\DeclareAcronym{OMP}{OMP}{orthogonal matching pursuit}
\DeclareAcronym{PT}{PT}{potential target}
\DeclareAcronym{PPP}{PPP}{Poisson point process}
\DeclareAcronym{RFS}{RFS}{random finite sets}
\DeclareAcronym{RMSE}{RMSE}{root mean square error}
\DeclareAcronym{ToA}{ToA}{time of arrival}
\DeclareAcronym{UID}{UID}{unique identifier}

\acrodefplural{BS}[BSs]{base stations}
\acrodefplural{ET}[ETs]{\revGrammar{extended targets}}
\acrodefplural{FOV}[FoVs]{fields of view}
\acrodefplural{PT}[PTs]{potential targets}

%% file: sections/sec_introduction.tex
\acresetall
\vspace{-5mm}
\section{Introduction}
\label{sec:introduction}
\Ac{ISAC} \cite{ISAC} uses a common wireless infrastructure for communication and monostatic, bistatic, or multistatic sensing \cite{meng2024cooperativesensing,meng2025cooperativeisacscaling}. Large deployments lead to \ac{DISAC} networks of \acp{BS} with partially overlapping \acp{FOV} \cite{strinati2024distributedintelligentintegratedsensing,strinati2024disacapproach} (Fig.~\ref{fig:target-handover-overview}). As infrastructure, monitored area, and target population grow, computation at every processing node and communication on every backhaul link must remain bounded with network size. Scalability is therefore a primary design objective.
\begin{figure}
    \centering
    \includegraphics[width=0.8\linewidth]{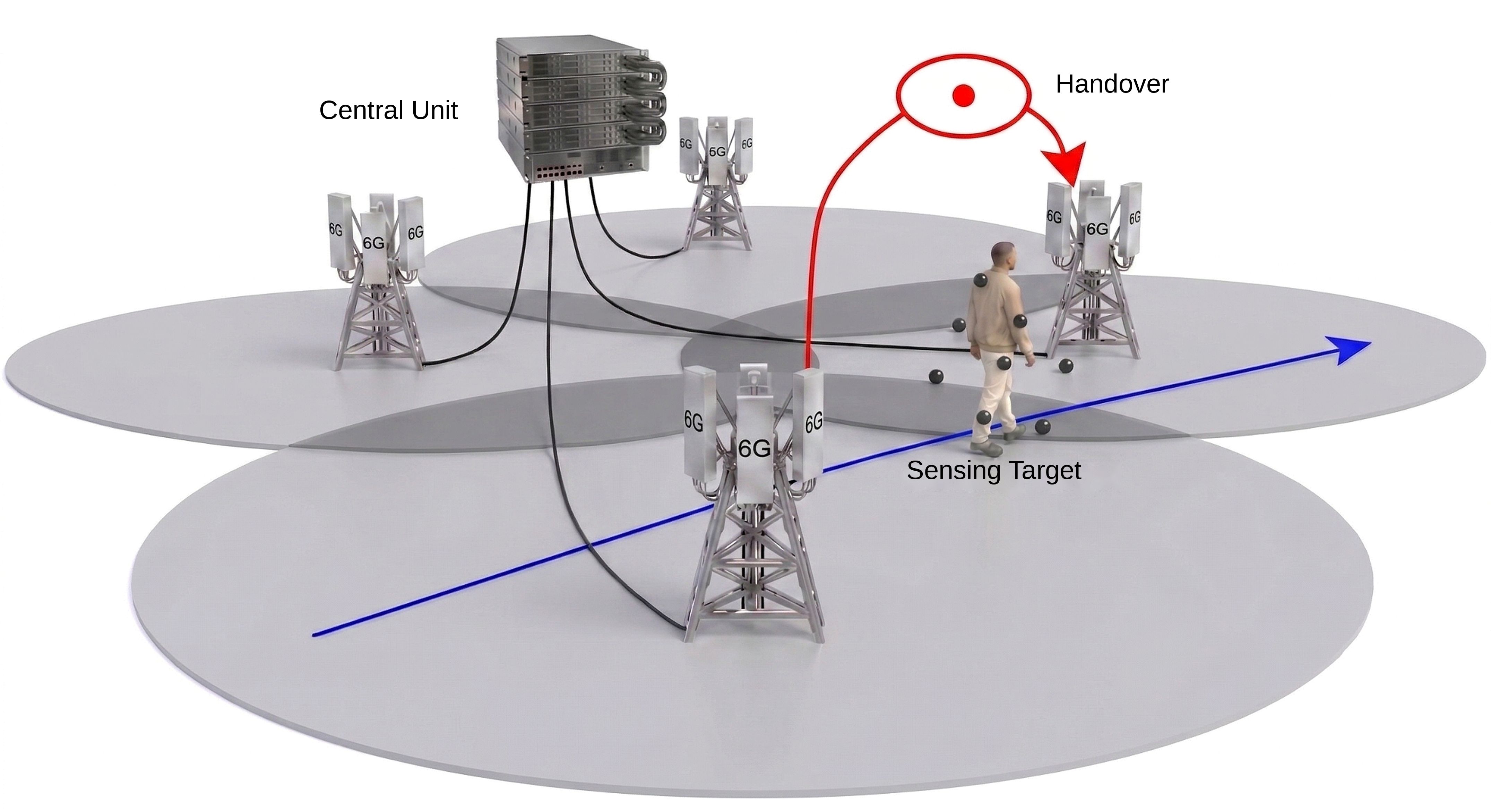}
    \caption{Extended-target handover in a \ac{DISAC} network. Gray discs
denote the \acp{FOV}, the blue curve is the target trajectory, black dots are
\revYu{incidence points}, and the red segment highlights handover between two \acp{BS}.\revGrammar{ The} central unit is for the coordinated baseline.}
    \label{fig:target-handover-overview}
\end{figure}

Scalability is required first at the local-filter level. The point-target model
used in most of the multisensor \ac{MTT} methods is inadequate when one physical target can
generate several spatially distributed detections in a scan. Such targets are
modeled as \acp{ET} \cite{granstrom2017extendedobjecttrackingintroduction}. In addition to
kinematic-state and data-association uncertainty, \revYu{\ac{ET} tracking} must account
for \emph{partition uncertainty}: the measurements are divided into groups,
and the resulting groups are associated with existing or newborn
targets. Existing methods marginalize over selected partitions
\cite{granstrom2012gmphd,lundquist2013cphd,granstrom2012phd,beard2016labeledeot},
retain multiple partition--association hypotheses in a multi-object posterior
\cite{granstrom2020pmbm,xia2022pmbapprox}, sample the joint partition and
association space \cite{sampling_method,xia2023trajectorybp}, or instantiate a
potential new target from each measurement
\cite{meyer2020scalable,meyer2021scalable}. These methods provide sophisticated
\revYu{\ac{ET} tracking} solutions, but they do not treat scalability as a
primary design objective.

Scalability is also required at the multisensor-fusion level. The established
multisensor \ac{MTT} literature provides a sound theoretical foundation for
fusion of \acp{PDF} \cite{fusion}.
Centralized multisensor filters combine measurements or likelihoods through
multi-scan estimation \cite{moratuwage2022multi}, \ac{CPHD} filtering
\cite{nannuru2016multisensor}, \ac{GLMB} filtering
\cite{vo2019multisensorglmb}, or \ac{PMB} filtering
\cite{frohle2019multisensorpmb}.
Centralized sum-product inference also supports joint cooperative localization
and point-target \ac{MTT}
\cite{BraGagSolRicGabLepNicWilBraWin:JSP2022}.
Distributed methods fuse local \ac{PHD} intensities using exponential-mixture
rules \cite{uney2013distributed}, \ac{CPHD} intensities using consensus
\cite{battistelli2013consensus_cphd}, or multi-view \ac{CPHD} intensities using
\ac{GCI} or arithmetic averaging \cite{li2021distributedmultiviewcphd}.
{Labeled multi-Bernoulli densities are fused using Kullback--Leibler averaging
\cite{fantacci2018distributed}}, while $\delta$-\ac{GLMB} densities are fused
using cross-entropy rules {\cite{saucan2018distributed}}.
For extended targets, decentralized \ac{PMB} fusion exchanges posterior
parameters describing existence, kinematics, and extent and combines the
corresponding densities through Kullback--Leibler averaging
\cite{frohle2020decentralizedpmb}. \ac{PMB} posteriors are also fused using
\ac{GCI} \cite{garcia2025distributed}.
Further extensions explicitly address limited or heterogeneous \acp{FOV}
\cite{van2021distributed,chen2025distributed,xiong2024distributedlimitedfov}.
For joint cooperative localization and point-target \ac{MTT}, decentralized
\ac{BP} exchanges agent-state beliefs and likelihood parameters while keeping target measurements local \cite{SharmaTSP2019}.
These methods form strong coordinated-processing baselines, but many assume
a fixed sensor network, repeated consensus, a global label structure, or a
fusion node whose load grows with the network. They therefore do not, in
general, establish the bounded per-node scalability required by \ac{DISAC}
networks.

Target handover offers a local alternative to continuous network-wide fusion.
Related ideas appear in multicamera \ac{MTT}, where tracks are transferred
across limited or non-overlapping views using predicted motion and appearance cues \cite{amosa2023multicamera,wang2022camerahandoff,tesfaye2019multitarget}.
Radio sensing cannot rely on visual appearance. Therefore, what \revGrammar{is} to be \revGrammar{handed over} in a \ac{DISAC} network needs to be studied. Recent
\ac{DISAC} and networked-radar studies have considered trajectory-\ac{PMB}
handover \cite{ge2024targethandoverdistributedintegrated}, \ac{BP}-based
point-target handover \cite{bai2025beliefpropagationbasedtargethandover}, target
assignment with seamless radar handovers \cite{kim2019optimal}, and local
handover protocols \cite{fan_handover}. These formulations either use
point-target models, rely on global track labels, or do not establish bounded
per-node computation and communication for growing networks. The closest prior works are summarized in \Tabref{tab:related_work}.

\begin{table}[!t]
\centering
\caption{Positioning of the proposed method relative to the closest handover
and scalable-sensing studies. \revYu{``Obj.'' identifies
the target model. ``Labels'' distinguishes local from global label management.
``Exch.'' identifies the exchanged prior (P), measurements (M), track-level
belief (Q), assignment (Assign.), or sensing configuration (Cfg.). ``OT''
denotes explicit ownership transfer. ``Scal.'' records whether bounded
per-node load as \(N\to\infty\) is established. N/R means not reported.}}
\label{tab:related_work}
{\footnotesize
\setlength{\tabcolsep}{3.5pt}
\begin{tabular}{@{}lccccc@{}}
\toprule
Work & Obj. & Labels & Exch. & OT & Scal.\\
\midrule
\cite{kim2019optimal} & {Point} & -- & Assign. & Yes & No\\
\cite{ribeiro2025mobilitymanagementintegratedsensing} & {Point} & -- & Cfg. & -- & N/R\\
\cite{ge2024targethandoverdistributedintegrated} & {Point} & Global & P & -- & No\\
\cite{bai2025beliefpropagationbasedtargethandover} & {Point} & Local & P,\,M & No & N/R\\
\cite{fan_handover} & {Point} & Local & P,\,M & Yes & N/R\\
\cite{liesegang2026scalableisac} & {Point} & -- & -- & -- & Yes\\
\textbf{Proposed} & {Extended} & Local & P\,/\,M\,/\,Q & Yes & Yes\\
\bottomrule
\end{tabular}}
\end{table}

\revGrBP{We propose a grouped-measurement \ac{BP} \ac{MTT} method, referred to as \grbp{}, together with event-triggered handover for \ac{DISAC}. Throughout the paper, \grbp{} denotes the common local \ac{MTT} method, while suffixes distinguish coordinated, uncoordinated distributed, and handover deployments. The proposed prior-only handover scheme is denoted \grbpv{H}. \grbpv{HM}, \grbpv{HL}, and \grbpv{HP} additionally exchange measurements, likelihood messages, and posteriors, respectively.}\acused{GrBP} Each \ac{BS} conditions its local update on one externally generated
measurement partition. 
Across
\acp{BS}, a local owner--shadow label structure transfers information only to
sensing neighbors for which handover is triggered. This \revGrammar{extends} the \ac{BP}-based handover method proposed in 
\cite{bai2025beliefpropagationbasedtargethandover} by introducing the \ac{ET}
measurement and filtering model, owner--shadow track management, additional
belief-exchange variants, and a formal per-node scalability analysis. The main
contributions are:
\begin{itemize}
\item \textbf{A range-dependent \ac{ET} detection scatter density model for \ac{DISAC} networks.} The model links the expected target-originated detection density to the radar resolution cell and the target-\ac{BS} range. 

    \item \textbf{\revGrBP{GrBP} based on grouped measurements under a single partition.} To deal with a potentially large number of measurements, the measurements at each scan are first converted into disjoint groups by an external clustering step, and \revGrBP{GrBP} then performs data association between legacy/newborn targets and these measurement groups. 
    
    \item \textbf{Proposed event-triggered handover for \revGrBP{GrBP}.} When the
confidence-gated expected detection count at a neighboring \ac{BS} exceeds a
threshold, the owner transfers the predicted track once (\grbpv{H}). The extensions exchange a measurement group (\grbpv{HM}), likelihood message (\grbpv{HL}), or posterior belief (\grbpv{HP}) during local co-observation.
\item \textbf{Scalability analysis.} We define a scalability criterion for \ac{DISAC} \ac{MTT}. Under the stated bounded-degree, bounded-local-workload, and finite-representation assumptions, centralized coordinated processing is not scalable, whereas the proposed handover variants and uncoordinated distributed processing have network-size-independent expected per-node loads.
\end{itemize}
\revGrBP{We assess information exchange and compare uncoordinated distributed GrBP with the state-of-the-art \ac{BP}-based \ac{ETT} method \cite{meyer2021scalable} in accuracy and runtime.} Code and detailed results are at \url{https://github.com/BaiLiping/EO_Target_Handover}. \secref{sec:system_model} gives the system model, and \secref{sec:detection_rate} gives the measurement model. \secref{sec:processing_architectures} presents \revGrBP{GrBP} and baseline deployments, and \secref{sec:handover} presents the proposed handover and scalability analysis. \secref{sec:numerical_results} presents numerical results and limitations, and \secref{sec:conclusion} concludes.

{\footnotesize\bl{\textit{Notation and Terminology:}
Bold lowercase and uppercase letters denote vectors and matrices. \((\cdot)^\T\),
\(\|\cdot\|_2\), \(|\cdot|\), and \(\ind\{\cdot\}\) denote transpose, Euclidean
norm, cardinality (or area for a planar region), and the indicator function.
\(\Set{N}(\cdot\mathpunct{,}\V m,\M P)\) denotes a Gaussian density. \revHyowon{\emph{Prior}} and \revHyowon{\emph{posterior}}
refer to \acp{PDF}. \revHyowon{Indices} are added as needed.}}

%% file: sections/sec_system_model.tex
\vspace{-1mm}
\section{System Model}
\label{sec:system_model}
\begingroup
\setlength{\abovedisplayskip}{3pt plus 1pt minus 1pt}
\setlength{\belowdisplayskip}{3pt plus 1pt minus 1pt}
\setlength{\abovedisplayshortskip}{2pt plus 1pt minus 1pt}
\setlength{\belowdisplayshortskip}{2pt plus 1pt minus 1pt}
\setlength{\jot}{2pt}

We consider a \ac{DISAC} network with a large number \(N\) of time-synchronized monostatic sensing\footnote{This work focuses exclusively on monostatic sensing in DISAC networks. User communication is not explicitly accounted for.} \acp{BS} that observe an unknown and time-varying number of extended targets over partially overlapping \acp{FOV}. New targets may enter the collective surveillance region, and existing targets may leave it. This section defines the network geometry, target dynamics, generic measurement model, and the scalability criterion.
\vspace{-1mm}
\subsection{Connectivity Graph}

\bl{Let \(\Set{B}\triangleq\{1,\ldots,N\}\). \ac{BS}~\(s\) is located at
\(\V{b}_s\in\mathbb R^2\), has \revGrammar{a} bounded sensing region
\(\Set{S}_s\subset\mathbb R^2\), and the collective surveillance region is
\(\Set{S}\triangleq\bigcup_{s\in\Set{B}}\Set{S}_s\). For distinct \acp{BS}
\(s,s'\), define their overlap and its area as
\(\Set{O}_{s,s'}\triangleq\Set{S}_s\cap\Set{S}_{s'}\) and
\(A_{s,s'}\triangleq\operatorname{area}(\Set{O}_{s,s'})\). The undirected
sensing-connectivity graph is \(G=(\Set{B},\Set{E})\), where
\(\Set{E}\triangleq\{\{s,s'\}:s,s'\in\Set{B},\ s<s',\ A_{s,s'}>0\}\).
\(s<s'\) lists each edge once. The sensing neighbors of \(s\) are
\(\Set{B}_s\triangleq\{s'\in\Set{B}:\{s,s'\}\in\Set{E}\}\).}
Each \ac{BS} knows its own sensing region, the sensing regions of its neighbors,
and the relative geometry required to evaluate pairwise overlap. Every sensing
edge is assumed to have a direct, reliable backhaul link within a sensing
frame. These idealized frame-level synchronization and delivery assumptions let
us isolate algorithmic computation and communication scaling. Delay, packet
loss, reordering, and clock offsets are left for future studies.

\subsection{Dynamics Model}

\bl{At time \(k\), target \(i\) has position \(\V{p}_k^i\in\mathbb R^2\),
kinematic state
\(\V{x}_k^i=[(\V{p}_k^i)^\T,(\V{\chi}_k^i)^\T]^\T\), where \(\V{\chi}_k^i\)
collects any additional model-dependent variables, and extent
\(\M{E}_k^i\in\mathbb S_{++}^2\). Its state is the kinematic--extent pair
\((\V{x}_k^i,\M{E}_k^i)\), and targets evolve independently.}
%
The single-object transition density is
\input{Equations/eq_state_transition}\bl{where the second equality assumes that,
given the current kinematic state, the next kinematic state is independent of
the current extent, and that the next extent is independent of the next
kinematic state given the current kinematic--extent pair.}
\vspace{-1mm}
\subsection{Measurement Model}

\bl{We next define the local measurement vector and single-scatter model.}

\subsubsection{{Measurements and Clutter}}

{At time step \(k\), \ac{BS}~\(s\) acquires \(m_{s,k}\) measurements, collected in the local measurement set}
$\V{Z}_{s,k}\triangleq
\{\V{z}_{s,k}^{1},\dots,\V{z}_{s,k}^{m_{s,k}}\}$.
{Each measurement can originate from either a target or clutter. Clutter is modeled by a \ac{PPP} with intensity \(c(\V{z})\), where the dependence on the receiving \ac{BS} is kept implicit unless needed.}

\subsubsection{{Single-Measurement Model}}

\bl{If \(\V{z}_{s,k}^j\) originates from target \(i\), then, following
\cite{meyer2021scalable},
\(\V{z}_{s,k}^j=\V{h}_s(\V{p}_k^i,\RV{o}_{s,k}^j)+\V{\nu}_{s,k}^j\), where
\(\V{h}_s\) is the nonlinear measurement function and \(\RV{o}_{s,k}^j\) is
the scattering-point offset. The models
\(\V{\nu}_{s,k}^j\sim\Set{N}(\V{0},\M{R}_s)\) and
\(\RV{o}_{s,k}^j\sim\Set{N}(\V{0},\zeta\M{E}_k^i)\) use measurement covariance
\(\M{R}_s\) and extent-scaling factor \(\zeta>0\)
\cite[Eq.~(13)]{granstrom2017extendedobjecttrackingintroduction}. Numerically,
\(\M{R}_s=\operatorname{diag}(\sigma_r^2,\sigma_\theta^2)\). Measurements may
contain channel parameters such as range and angle, obtained using standard
estimators \cite{venugopal2017channel,jiang2021beamspace}.}
{Separately, choose a footprint scale \(\kappa>0\) for the bounded extent
used by the resolution-cell count and visibility model, and define}
\(
{\mathcal{A}_\kappa(\V{x}_k^i, \M{E}_k^i)}
\triangleq
\left\{
\V{q}\in\mathbb{R}^2 :
(\V{q}-\V{p}_k^i)^\T
(\M{E}_k^i)^{-1}
(\V{q}-\V{p}_k^i)
\le \kappa
\right\},
\)
{with area}
\(
A_\kappa(\M{E}_k^i)=\operatorname{area}(\mathcal{A}_\kappa(\V{x}_k^i, \M{E}_k^i))=\pi\kappa\sqrt{\det\M{E}_k^i} .
\)
{The parameters \(\zeta\) and \(\kappa\) therefore serve different
purposes. \(\kappa\) is not interpreted as a Gaussian confidence probability.}

{The state-dependent Poisson measurement-count mean and the marginalized
single-measurement likelihood, including finite-\ac{FOV} visibility, are specified
together in \secref{sec:detection_rate}.}
\vspace{-1mm}
\subsubsection{Measurement Grouping}

All \grbp{}-based architectures operate on grouped measurements obtained from an external clustering algorithm (e.g., \ac{DBSCAN} \cite{ester1996density}), rather than associating individual detections directly. This grouping step belongs to \revGrBP{GrBP}. \revGrBP{The BP-based ETT method instead associates individual measurements (\secref{sec:numerical_results}).} At each scan, the external clustering algorithm partitions the local measurement set \(\V{Z}_{s,k}\) into \(n_{s,k}^{g}\) disjoint groups,
$\V{Z}_{s,k}^{G} \triangleq
\{\V{Z}_{s,k}^{1},\dots,\V{Z}_{s,k}^{n_{s,k}^{g}}\}$,
where \(\V{Z}_{s,k}^{j}\) denotes the \(j\)th group. Each group is treated as originating from at most one target. 

\subsection{Scalability Criterion and Assumptions}
\label{subsec:problem_formulation}

{We study scalability in the regime where the number of \acp{BS} \(N\) and the number of targets grow unboundedly. Following \cite[Definition 1]{liesegang2026scalableisac}, we assume bounded sensing regions {and a deployment with uniformly bounded local \ac{BS} density}. {We separately require} a bounded number of sensing neighbors, with maximum sensing-graph degree \(d_{\max}\triangleq\sup_{N}\max_{s\in\Set{B}}|\Set{B}_s|\) independent of \(N\). A scalable \ac{MTT} management scheme must keep both local processing and inter-\ac{BS} communication bounded independently of \(N\). Since these two resources are affected by different mechanisms, we state separate computation and communication conditions.}
The scalability analysis rests on three regularity assumptions.

\begin{assumption}[Local interaction]
\label{as:degree}
The sensing-graph degree is bounded, \(|\Set{B}_s|\le d_{\max}<\infty\)
uniformly in \(s\) and \(N\). 
\end{assumption}

To characterize the local workload, let $n^t_{s,k}$ denote the
number of retained legacy targets at \ac{BS}~\(s\) that enter the
frame-$k$ local update. Let
$n^g_{s,k}\triangleq |\V{Z}_{s,k}^{G}|$ denote the number of
local measurement groups. Thus, the local update processes
$n^t_{s,k}$ legacy targets and introduces at most one newborn
component per measurement group.

\begin{assumption}[Stable local workload]
\label{as:workload}
The numbers of legacy targets and measurement groups have uniformly
bounded second moments: there exist finite constants
$\bar n^t$ and $\bar n^g$, independent of $N$, such that $\sup_{N,s,k}\mathbb{E}\![(n^t_{s,k})^2]
\leq \bar n^t$ and $\sup_{N,s,k}\mathbb{E}\![(n^g_{s,k})^2]
\leq \bar n^g$.
\end{assumption}

Assumption \ref{as:workload} holds under a standard stationary marked-\ac{PPP} target model with bounded local sensing footprints and count moments, \ac{PPP} clutter of bounded local intensity, partition-based grouping, and stable track management with finite unsupported-track lifetime.

\begin{assumption}[Finite representation]
\label{as:representation}
Every belief and track-level message in the MTT filter is represented with at most a fixed, \(N\)-independent number of parameters (e.g., a bounded \ac{EKF}/Gaussian--inverse-Wishart mixture order), the single-\ac{BS} update runs a bounded number of message-passing iterations, and each handover criterion is evaluated at \(N\)-independent cost. {Local grouped-measurement likelihoods are represented by fixed-dimensional sufficient statistics. Event/control and track-level packet sizes are bounded by \(N\)-independent constants. For measurement handover, let \(B^{\max}_{Z,s,k}\) denote the largest encoded byte size of a local raw measurement group in frame \(k\). We assume \(\sup_{N,s,k}\mathbb E[(B^{\max}_{Z,s,k})^2]<\infty\), rather than a deterministic size cap.}
\end{assumption}


A \ac{DISAC} management scheme $\mathcal{M}$ specifies the set of processing nodes $\Set{V}(\mathcal{M})$ (or $\Set{V}_{N}(\mathcal{M})$ to make $N$ explicit), with \(\Set{V}(\mathcal M)=\Set{B}\) for fully distributed
schemes, and \(\Set{B}\cup\{0\}\) when a central fusion node \(0\) is used. Let
\(C^{\mathrm{comp}}_{v,k}(\mathcal M)\) and \(C^{\mathrm{comm}}_{v,k}(\mathcal M)\)
denote the computation and the total data sent and received at node \(v\) in
frame~\(k\) under management scheme $\mathcal{M}$. 

\begin{definition}[{Scalable \ac{DISAC} \ac{MTT}}]
\label{def:scalable}
A management scheme \(\mathcal M\) is scalable if there exist finite constants \(\bar C^{\mathrm{comp}}\) and \(\bar C^{\mathrm{comm}}\), independent of \(N\), such that
\begin{align}
\sup_{N}\max_{v\in\Set{V}_N(\mathcal M)}\sup_{k}
\mathbb E\!\left[
C^{\mathrm{comp}}_{v,k}(\mathcal M)
\right]
&\leq \bar C^{\mathrm{comp}},
\label{eq:scalable_comp}
\\
\sup_{N}\max_{v\in\Set{V}_N(\mathcal M)}
\limsup_{K\to\infty}\frac{1}{K}
\sum_{k=1}^{K}
\mathbb E\!\left[
C^{\mathrm{comm}}_{v,k}(\mathcal M)
\right]
&\leq \bar C^{\mathrm{comm}}.
\label{eq:scalable_comm}
\end{align}
\end{definition}

Condition~\eqref{eq:scalable_comp} bounds the per-frame computation at every node, including any fusion node, so a scheme is non-scalable as soon as some node must process a population that grows with \(N\). Condition~\eqref{eq:scalable_comm} bounds each node's long-run communication, counting both transmitted and received data. 
By accounting for the fusion node and counting incoming as well as outgoing data, the definition captures architectures in which the per-\ac{BS} load is small but a central node becomes the bottleneck. The objective is therefore a scheme that satisfies Definition~\ref{def:scalable} while allowing useful neighbor-to-neighbor exchange.
\endgroup

%% file: Equations/eq_state_transition.tex
\begingroup
\begin{align}
&f(\V{x}_{k+1}^i,\M{E}_{k+1}^i
  \mid \V{x}_{k}^i,\M{E}_{k}^i)
  \nonumber \\
&\quad =
f(\V{x}_{k+1}^i\mid\V{x}_k^i,\M{E}_{k}^i)
f(\M{E}_{k+1}^i\mid\V{x}_{k+1}^i,\V{x}_k^i,\M{E}_{k}^i)
\nonumber \\
&\quad =
f(\V{x}_{k+1}^i\mid\V{x}_k^i)
f(\M{E}_{k+1}^i\mid\V{x}_k^i,\M{E}_{k}^i).
\label{eq:state_transition}
\end{align}
\endgroup

%% file: sections/sec_detection_density.tex
\vspace{-1mm}

\section{{Resolution- and Visibility-Aware Extended-Target Measurement Model}}
\label{sec:detection_rate}

\begingroup

This section refines the generic measurement model in \secref{sec:system_model}
by specifying the expected number of measurements generated by an extended
target. The model combines two effects: the number of resolvable scattering
cells decreases with distance through a scatter density \(\rho_s(\V{x})\), and
only part of the target extent may be effectively visible to \ac{BS}~\(s\)
through a \ac{BS}-side visibility response \(\eta_s(\V{q})\).
\vspace{-1mm}
\subsection{Resolution-Based Scatter Density}

Consider a monostatic sensing \ac{BS}~\(s\) located at \(\V{b}_s\). 
Let \(c_0\)
denote the speed of light, \(B\) the signal bandwidth, \(\lambda_{\mathrm c}\)
the carrier wavelength, and \(D_{\mathrm T}\) and \(D_{\mathrm R}\) the transmit
and receive array apertures, respectively. The virtual-aperture length is
\(
D_{\mathrm V}=D_{\mathrm T}+D_{\mathrm R}.
\)
For a target state \(\V{x}\) with position component \(\V{p}\), the
target-to-\ac{BS} distance is \(r_s(\V{x})\triangleq\|\V{p}-\V{b}_s\|_2\). Under
the standard matched-filter approximation, the range and cross-range resolution
widths are \(\Delta_{\mathrm r}=c_0/(2B)\) and
\(\Delta_{\mathrm{cr},s}(\V{x})=r_s(\V{x})\lambda_{\mathrm c}/D_{\mathrm V}\),
giving the two-dimensional resolution-cell area
\(A_{\mathrm{cell},s}(\V{x})\triangleq\Delta_{\mathrm r}\Delta_{\mathrm{cr},s}(\V{x})=c_0 r_s(\V{x})\lambda_{\mathrm c}/(2BD_{\mathrm V})\)
{\cite{richards2014fundamentals}}.
Ignoring residual resolution cells at the target boundary, the number of
resolution cells covered by the scaled extent is approximated by
\(N_{\mathrm{cell},s}(\V{x},\M{E})\approx A_\kappa(\M{E})/A_{\mathrm{cell},s}(\V{x})\), for a chosen value \({\kappa>0}\).
If each resolved cell produces at most one detected measurement with per-cell
detection probability \(p_{\mathrm d,s}^{\mathrm{cell}}\in[0,1]\), then the
expected number of measurements under full visibility can be written as
\(
\rho_s(\V{x})A_\kappa(\M{E})
\)
with
\begin{equation}
\rho_s(\V{x})
\triangleq
\frac{p_{\mathrm d,s}^{\mathrm{cell}}}{A_{\mathrm{cell},s}(\V{x})}
=
\frac{2BD_{\mathrm V}}{c_0\lambda_{\mathrm c}r_s(\V{x})}
\,p_{\mathrm d,s}^{\mathrm{cell}}.
\label{eq:rho_from_pd_cell_density}
\end{equation}
Thus, \(\rho_s(\V{x})\) is an effective spatial scatter density that includes
the resolution-cell geometry and the per-cell detection probability.
{It has units of expected detections per square metre. It is an intensity,
not a normalized probability density, and therefore need not integrate to one.}
Because \(\rho_s(\V{x})\propto 1/r_s(\V{x})\), the per-target detection mean
grows without bound only as a target approaches the \ac{BS}. Hence, scalability
Assumption~\ref{as:workload} requires a non-zero minimum target-to-\ac{BS}
range \(r_{\min}>0\). 

\begin{remark}[Two uses of the extent]
\bl{The bounded ellipse \(\mathcal A_\kappa\) is used to count resolvable
cells, whereas \(\Set{N}(\V{0},\zeta\M{E})\) models the locations of
target-generated measurements. Visibility therefore weights the expected count
by geometric area but conditions the measurement likelihood by Gaussian mass.
The two representations agree under full visibility but can differ near a hard
\ac{FOV} boundary.}
\end{remark}
\vspace{-1mm}
\subsection{Visibility-Weighted Extent Fraction}
\label{subsec:fov_visibility_model}

The scatter density \(\rho_s(\V{x})\) accounts for distance-dependent
resolution, but it does not describe partial visibility of an extended target.
We model \ac{BS}-side visibility by a response function
\(\eta_s(\V{q})\in[0,1]\), where \(\V{q}\in\mathbb{R}^2\) is a point on the target
extent. The function \(\eta_s\) \revYu{may take} antenna gain, blockage, boundary
loss, or any other spatial variation in sensing quality \revYu{into consideration}. {A hard-\ac{FOV}
model is the special case}
\(
{\eta_s^{\mathrm{hard}}(\V{q})=\ind\!\left\{\V{q}\in\Set{S}_s\right\}.}
\label{eq:hard_fov_response}
\)
For an extended target with state \((\V{x},\M{E})\), define the
visibility-weighted extent fraction
\begin{equation}
\delta_s(\V{x},\M{E})
\triangleq
\frac{1}{A_\kappa(\M{E})}
\int_{\mathbb{R}^2}
\eta_s(\V{q})
\ind_{\mathcal{A}_\kappa(\V{x},\M{E})}(\V{q})
\,\mathrm d\V{q}
\in[0,1].
\label{eq:system_visible_scattering_mass}
\end{equation}
When \(\eta_s(\V{q})=1\) over the whole scaled extent,
\(\delta_s(\V{x},\M{E})=1\). When a hard-\ac{FOV} response is used, \(\delta_s\) is the fraction of the
scaled extent lying inside the sensing region \(\Set{S}_s\). {In particular,
\(\delta_s=0\) when the scaled extent has no overlap with \(\Set S_s\), and
\(\delta_s=1\) when it lies wholly inside \(\Set S_s\).}
\vspace{-1mm}
\subsection{Measurement-Count Mean and Likelihood}

Combining the distance-dependent scatter density and the visibility-weighted
extent fraction gives the expected number of target-generated measurements at
\ac{BS}~\(s\):
\begin{equation}
\mu_{m,s}(\V{x},\M{E})
\triangleq
\rho_s(\V{x})
A_\kappa(\M{E})
\delta_s(\V{x},\M{E}).
\label{eq:system_measurement_count_mean}
\end{equation}
This expression reduces to the full-visibility resolution-cell model
\(\rho_s(\V{x})A_\kappa(\M{E})\) when \(\delta_s(\V{x},\M{E})=1\).
The target-generated count is modeled as Poisson with this mean. For \(r_s\ge r_{\min}\) the mean is finite and
\(N\)-independent, so the count has finite moments despite its unbounded support. 
{Given that a measurement was generated by the target, the unweighted
single-measurement likelihood follows \cite[Eq.~(6)]{meyer2021scalable} and
marginalizes the unknown scattering offset,}
\begin{equation}
{f_s(\V z\mid\V x,\M E)
=\int f_s(\V z\mid\V p,\V o)f(\V o\mid\M E)\,\mathrm d\V o.}
\label{eq:measurement_likelihood}
\end{equation}
The same visibility response can also be used to condition this likelihood on
the visible part of the target. Generalizing
\cite[Eq.~(6)]{meyer2021scalable}, we write
\begin{equation}
f_s^{\eta}(\V{z}\mid\V{x},\M{E})
=
\frac{
\int
\eta_s(\V{p}+\V{o})
f_s(\V{z}\mid\V{p},\V{o})
f(\V{o}\mid\M{E})
\,\mathrm d\V{o}
}{
\int
\eta_s(\V{p}+\V{o})
f(\V{o}\mid\M{E})
\,\mathrm d\V{o}
}.
\label{eq:visibility_conditioned_measurement_likelihood}
\end{equation}
The unweighted likelihood in \eqref{eq:measurement_likelihood} is recovered
when \(\eta_s(\V{q})=1\) over the target extent. With the hard-\ac{FOV}
response, \eqref{eq:visibility_conditioned_measurement_likelihood} conditions
the likelihood on scattering points that lie inside \(\Set{S}_s\). The
conditioned likelihood is used only when the denominator in
\eqref{eq:visibility_conditioned_measurement_likelihood} is positive. If the
{bounded-footprint overlap in \eqref{eq:system_visible_scattering_mass}
is zero, the count mean in \eqref{eq:system_measurement_count_mean} is set to
zero and no target-originated likelihood contribution is used, even though the
unbounded Gaussian offset model can assign an arbitrarily small tail mass inside
the \ac{FOV}.}

\endgroup

%% file: sections/sec_reference_methods.tex
\vspace{-1mm}
\section{\revGrBP{GrBP and Multi-BS Processing Architectures}}
\label{sec:processing_architectures}

This section presents \revGrBP{GrBP} and its uncoordinated and coordinated baseline deployments for the proposed handover methods (\secref{sec:handover}). A \emph{potential target} is a track component whose existence, kinematics, and extent are jointly inferred. Legacy components are retained from the previous frame and newborn components originate from current measurement groups. For local scalability, \revGrBP{GrBP} conditions on one partition (\secref{sec:system_model}) and infers group-to-target associations, unlike filters introducing a potential target per measurement \cite{meyer2020scalable,meyer2021scalable}. At a fixed frame \(k\), we describe each architecture's prior, measurement usage, and posterior input/output relations.
\vspace{-1mm}
\subsection{\revGrBP{GrBP and Uncoordinated Distributed Deployment}}
\label{sec:single_bs_method}

\revGrBP{GrBP} builds on the \ac{BP}-based \ac{MTT} formulation in \cite{meyer2018messagepassing}. Each potential extended target is represented by the augmented state
\(\V{y}_k^i=[\V{x}_k^i,\M{E}_k^i,r_k^i]^\T\) with kinematic state \(\V{x}_k^i\),
extent matrix \(\M{E}_k^i\), and existence indicator \(r_k^i\in\{0,1\}\). In this \revYu{section}, we drop
the qualifier ``potential'' and the \ac{BS} index. The target states are collected into \(\V{Y}_k=(\underline{\V{Y}}_k,\overline{\V{Y}}_k)\) and split into the
\(n_k^{t}\) legacy targets \(\underline{\V{Y}}_k= \{ \V{y}_k^p\}_{p=1}^{n_k^{t}}\) carried from the previous
frame and the \(n_k^{g}\) \rev{newborn targets} \(\overline{\V{Y}}_k= \{ \overline{y}_k^p\}_{p=1}^{n_k^{g}}\), one per
measurement group.

\subsubsection{Prior information}\label{sec:single_bs_prior} At the start of frame \(k\), the \ac{BS} holds the following information: each legacy target carries its frame-\((k{-}1)\) posterior
belief \(\Tilde{f}(\underline{\V{y}}_{k-1}^{p})\), which through the augmented
state jointly describes kinematics, extent, and existence. \begin{revision}Propagating these
through the transition model \eqref{eq:state_transition} gives the predicted
transition-factor messages \(\alpha_k^p(\underline{\V{y}}_k^{p})\), collected as
\(\mathcal F_k^-\triangleq\{\alpha_k^p(\underline{\V{y}}_k^{p})\}_{p=1}^{n_k^{t}}\)
for use at frame \(k\).\end{revision}

\subsubsection{Measurement usage} The grouped-measurement vector \(\V{Z}_k^{G}\) defined
in \secref{sec:system_model} enters the update through the single-scan joint
density function \bl{derived in }
\blsuppref{sec:appendix_I}{Supplementary Section~S-A}\bl{, which factorizes as}
\begin{align}
&f(\underline{\V{Y}}_k,\overline{\V{Y}}_k,\V{a}_k,\V{b}_k,\V{Z}_k^{G})
\notag\\[-0.2ex]
&\quad\propto \boldsymbol{\psi}(\V{a}_k,\V{b}_k)
\prod_{p=1}^{n_k^t}
  \rev{\alpha_k^p(\underline{\V{y}}_k^{p})}\,
  \underline{l}(\underline{\V{y}}_k^{p},a_k^{p}\mathpunct{,}\V{Z}_k^{G}) \notag \\
&\times \prod_{j=1}^{n_k^{g}}
  \overline{l}(\overline{\V{y}}_k^{j}, b_k^{j}\mathpunct{,}\V{Z}_k^{j}).\label{eq:single_bs_joint_pdf_factorization}
\end{align}
Here \rev{\(\alpha_k^p(\underline{\V{y}}_k^{p})\) is the predicted transition-factor
message (equivalently, the predicted marginal) of the \(p\)th legacy target}, \(\underline{l}(\cdot)\) and \(\overline{l}(\cdot)\) are the
legacy and \rev{newborn} likelihood factors \bl{defined in }
\bleqref{eq:legacy_l_function}\bl{~and~}\bleqref{eq:new_l_function}\bl{, and the association}
factor \(\boldsymbol{\psi}(\V{a}_k,\V{b}_k)\) enforces that each measurement
group is assigned to at most one target and vice versa
\cite[Fig.~4]{meyer2018messagepassing}. Conditioning on the single partition
\(\V{Z}_k^{G}\) confines inference to the group-to-target association variables
\((\V{a}_k,\V{b}_k)\). This is what keeps the update tractable and distinguishes
\revGrBP{GrBP} from formulations that instantiate a new target per measurement
\cite{meyer2020scalable,meyer2021scalable}.

\subsubsection{Posterior computation} The marginal posteriors are approximated by beliefs
from sum-product message passing in the single-\ac{BS} update, using the schedule of
\cite[Sec.~IX.A]{meyer2018messagepassing}. \revHyowon{The belief associated with a variable node approximates the marginal \ac{PDF} of that variable in the factor graph.} \begin{revision}The association
weights, iterative association messages, marginal association beliefs,
legacy/newborn state beliefs, and existence beliefs are described in the references cited in
\blsuppref{sec:appendix_beliefs}{Supplementary Section~S-A.C}.\end{revision} For each legacy target, the
prior-factor message is the \rev{predicted transition-factor message
\(\alpha_k^i(\underline{\V{y}}_k^i)\)} and the likelihood-factor message
\(\gamma(\underline{\V{y}}_k^i)\) summarizes the grouped-measurement evidence, so
the posterior belief is
\(\Tilde{f}(\underline{\V{y}}_k^i)\propto
\rev{\alpha_k^i(\underline{\V{y}}_k^i)}\,\gamma(\underline{\V{y}}_k^i)\)
(and analogously for \rev{a newborn target}). As the augmented state includes
existence, this also yields the posterior existence probability. The likelihood
message \(\gamma(\cdot)\) is the only quantity needed by the multi-\ac{BS}
architectures below. The \rev{posterior-belief collection
\(\mathcal F_k^+\triangleq\{\Tilde{f}(\underline{\V{y}}_k^{p})\}_{p=1}^{n_k^t}\cup
\{\Tilde{f}(\overline{\V{y}}_k^{j})\}_{j=1}^{n_k^g}\) is propagated through the
transition model to form \(\mathcal F_{k+1}^-\), with each newborn target
treated as legacy}.

\subsubsection{Input/output relation} \leavevmode\revGrBP{GrBP} maps a local
prior and local grouped measurements to a local posterior, with no inter-\ac{BS}
exchange.

\subsubsection{Uncoordinated distributed processing} The simplest multi-\ac{BS}
architecture, denoted \grbpv{D}, runs an independent copy of \revGrBP{GrBP} at every \ac{BS}, with no information exchanged. A target in the overlap of several \acp{FOV}
is then represented by independent targets at the corresponding \acp{BS}, and a
target entering a new \ac{BS} must be reinitialized there even if a neighbor
already tracks it. The per-frame input/output is identical to the single-\ac{BS} case, with
zero inter-\ac{BS} communication.
\vspace{-1mm}
\subsection{Coordinated Multi-BS \grbp{} Processing}
\label{sec:coordinated_processing}

\revGrBP{Coordinated GrBP processing} combines observing \acp{BS}' information into one posterior per target, using sequential (\grbpv{CS}) or parallel (\grbpv{CP}) schedules of \revGrBP{GrBP}. Appending the \ac{BS} index \(s\), let \(\mathcal{U}_s\) denote the
single-\ac{BS} update of \secref{sec:single_bs_method} driven by
\(\V{Z}_{s,k}^{G}\): through its likelihood messages \(\gamma_s(\cdot)\) it refines the
beliefs of the legacy targets in the \ac{FOV} of \ac{BS}~\(s\) and leaves the rest
unchanged. All \acp{BS} share the same \rev{predicted legacy messages}
\begin{equation}
{\color{black}
\alpha_k(\underline{\V{Y}}_k)
\triangleq\prod_{p=1}^{n_k^t}\alpha_k^p(\underline{\V{y}}_k^p),
}
\label{eq:coordinated_prior}
\end{equation}
obtained as in \revHyowon{\secref{sec:single_bs_prior}}. The two schedules differ only in
how the per-\ac{BS} updates \(\{\mathcal{U}_s\}\) are combined, and each can be
realized with or without a central fusion node.

These are representative coordinated realizations. Sequential updates use the predecessor's belief and can refine its newborn targets, but are order-dependent and cannot run concurrently. Parallel updates share the same prediction and combine likelihoods afterwards. They are order-invariant and parallelizable. Our per-\ac{BS} newborn proposals can, however, duplicate targets unless an additional cross-\ac{BS} newborn association step is used.

\subsubsection{Sequential \grbp{} Processing (\grbpv{CS})}
\label{sec:sequential_processing}
Under sequential processing, the \acp{BS} are processed in a fixed
order \(1,\dots,N\).
 \paragraph*{Measurement usage and posterior} The posterior beliefs are obtained by composing the
per-\ac{BS} updates in the given order,
\begin{equation}
{\color{black}
\Tilde{f}_k(\underline{\V{Y}}_k)=
\big(\mathcal{U}_{N}\circ\cdots\circ\mathcal{U}_{1}\big)
\big[\alpha_k(\underline{\V{Y}}_k)\big].
}
\label{eq:sequential_update}
\end{equation}
Each update \(\mathcal{U}_s\) takes the beliefs produced by the preceding
\acp{BS} as its prior and refines only the targets in the \ac{FOV} of
\ac{BS}~\(s\). Targets outside that \ac{FOV} are passed through unchanged.

\paragraph*{Input/output relation} With a central fusion node, every \ac{BS} sends
\(\V{Z}_{s,k}^{G}\) to the node, which evaluates \eqref{eq:sequential_update}.
Without a central unit, the \acp{BS} share the \rev{predicted-message collection} and apply
\(\mathcal{U}_1,\dots,\mathcal{U}_N\) along a chain, with each \ac{BS} forwarding its updated
beliefs to the next \ac{BS}. Either way the update is a chain of \(N\) steps whose result
can depend on the processing order.

\subsubsection{Parallel \grbp{} Processing (\grbpv{CP})}
\label{sec:parallel_processing}
In contrast to sequential processing, here there is no predetermined order.
\paragraph*{Measurement usage and posterior} In parallel processing, every
\ac{BS} updates the shared prior \eqref{eq:coordinated_prior} independently, so
each observing \ac{BS} computes its likelihood message \(\gamma_s(\cdot)\) from
the \rev{same predicted transition-factor message}. For a target observed by the set
\(S_i\subseteq\Set{B}\) of \acp{BS}, which by the bounded sensing geometry of
\secref{subsec:problem_formulation} satisfies \(|S_i|\le d_{\max}+1\), the fused
posterior for a \emph{legacy} target is the product
\begin{equation}
{\color{black}
\Tilde{f}_k^i(\underline{\V{y}}_k^i)\;\propto\;
\alpha_k^i(\underline{\V{y}}_k^i)
\prod_{s\in S_i}\gamma_s(\underline{\V{y}}_k^i).
}
\label{eq:parallel_belief}
\end{equation}

\paragraph*{Input/output relation} The messages are computed independently and
combined in a single product, so there is no dependency chain and the result is
independent of the \ac{BS} ordering. In return, the likelihood messages
\(\{\gamma_s\}_{s\in S_i}\) of each target must be collected at the fusion node,
or at a designated combiner, before the posterior is formed.

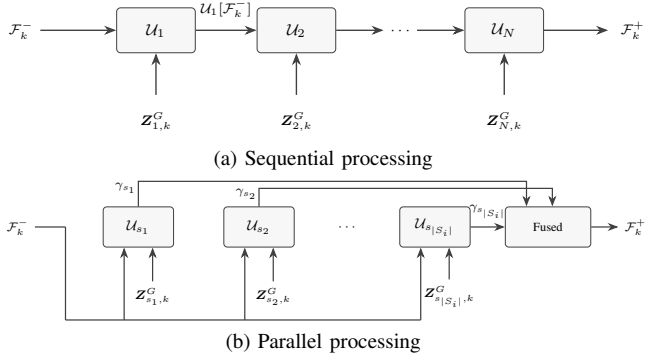
\begin{figure}
    \centering
    \captionsetup[subfloat]{position=bottom}
    \subfloat[Sequential processing\label{fig:sequential-processing}]{%
        \resizebox{\columnwidth}{!}{\input{Drawings/sequential_processing_io}}%
    }\\[0.2em]
    \subfloat[Parallel processing\label{fig:parallel-processing}]{%
        \resizebox{\columnwidth}{!}{\input{Drawings/parallel_processing_io}}%
    }
    \vspace{0.45em}
    \caption{Coordinated \grbp{} processing schedules. Subfigures (a) and (b) show \grbpv{CS} and \grbpv{CP}, respectively, for the \(i\)th target, which lies within the overlapping sensing region of \(|S_i|\) \acp{BS}.}
    \label{fig:coordinated-processing}
\end{figure}

\subsection{Scalability Analysis}
\label{sec:reference_scalability}

Although \revGrBP{GrBP} has bounded local cost under the stated assumptions, deployment scalability depends on where information is processed and exchanged (Definition~\ref{def:scalable}).

\subsubsection{Uncoordinated distributed processing} All inference is performed
locally and independently at each \ac{BS}. Under
Assumptions~\ref{as:degree}--\ref{as:representation}, the numbers of legacy
targets and measurement groups \(n_{s,k}^{t}\) and \(n_{s,k}^{g}\) have bounded
moments independent of \(N\) (\secref{subsec:problem_formulation}), and each component is
updated at bounded cost, so the per-\ac{BS} cost
\begin{equation}
C^{\mathrm{comp}}_{s,k}=O\!\big(n_{s,k}^{t}\,n_{s,k}^{g}\big),
\qquad
\mathbb E\!\big[C^{\mathrm{comp}}_{s,k}\big]=O(1)
\label{eq:local_comp_cost}
\end{equation}
is bounded in expectation independently of \(N\), and no messages are
exchanged, \(C^{\mathrm{comm}}_{s,k}=0\). Here \(\Set{V}(\mathcal M)=\Set{B}\), as
there is no fusion node. \grbpv{D} therefore satisfies both conditions of Definition~\ref{def:scalable}.

\subsubsection{Coordinated processing} Fusing several \acp{BS}' measurements into one
posterior couples their data, so the joint update cannot be split into independent
bounded-size local updates. In the centralized realization the \acp{BS} merely
forward their measurements and all inference runs at the fusion node
\(0\in\Set{V}(\mathcal M)\), whose computation and communication grow with \(N\). Since
Definition~\ref{def:scalable} bounds every node, including node~\(0\), this is
the bottleneck, as formalized next. \footnote{For sequential processing, the central node can be removed, but that
does not restore scalability. Depending on how the sequential processing is implemented, the final node might need to process \revGrammar{an} infinite number of tracks.}

\begin{proposition}[Non-scalability of centralized coordinated processing]
\label{prop:centralized_nonscalable}
Consider the coordinated architectures of \secref{sec:coordinated_processing}
realized with a central fusion node \(0\), under the bounded assumptions of
\secref{subsec:problem_formulation}. Assume the expected number of measurements
per \ac{BS} per frame is at least \(\mu_{\min}>0\). Then the fusion node's expected
per-frame computational and communication loads, \(C^{\mathrm{comp}}_{0,k}\) and
\(C^{\mathrm{comm}}_{0,k}\), both grow at least linearly in \(N\). Consequently, the centralized coordinated
architectures are not scalable in the sense of Definition~\ref{def:scalable}.
\end{proposition}
\begin{proof}
    See Appendix~\ref{sec:appendix_scalability}.
\end{proof}


%% file: Drawings/sequential_processing_io.tex
\definecolor{seqBox}{HTML}{F7F7F7}
\definecolor{seqText}{HTML}{000000}

\begin{tikzpicture}[
  >=Stealth,
  font=\footnotesize,
  updatebox/.style={
    draw=black!75,
    rounded corners=2pt,
    fill=seqBox,
    line width=0.45pt,
    align=center,
    inner sep=2.5pt,
    minimum width=14mm,
    minimum height=8mm
  },
  flow/.style={->, line width=0.7pt, draw=black!75},
  data/.style={font=\scriptsize, align=center, text=seqText}
]

\node[data] (prior) at (0,0)
  {\textcolor{black}{\(\mathcal F_k^-\)}};

\node[updatebox] (u1) at (2.4,0) {\(\mathcal U_1\)};
\node[updatebox] (u2) at (4.9,0) {\(\mathcal U_2\)};
\node[data] (dots) at (6.75,0) {\(\cdots\)};
\node[updatebox] (uN) at (8.6,0) {\(\mathcal U_N\)};

\coordinate (out) at (10.45,0);

\draw[flow] (prior) -- (u1);
\draw[flow] (u1) -- node[data, above=1.0mm] {\textcolor{black}{\(\mathcal U_1[\mathcal F_k^-]\)}} (u2);
\draw[flow] (u2) -- (dots);
\draw[flow] (dots) -- (uN);
\draw[flow] (uN) -- (out);

\node[data, anchor=west] at (10.55,0)
  {\textcolor{black}{\(\mathcal F_k^+\)}};

\coordinate (m1) at (2.4,-1.15);
\coordinate (m2) at (4.9,-1.15);
\coordinate (mN) at (8.6,-1.15);

\draw[flow] (m1) -- (u1);
\draw[flow] (m2) -- (u2);
\draw[flow] (mN) -- (uN);

\node[data] at (2.4,-1.62) {\(\V{Z}_{1,k}^{G}\)};
\node[data] at (4.9,-1.62) {\(\V{Z}_{2,k}^{G}\)};
\node[data] at (8.6,-1.62) {\(\V{Z}_{N,k}^{G}\)};

\end{tikzpicture}

%% file: Drawings/parallel_processing_io.tex
\definecolor{parBox}{HTML}{F7F7F7}
\definecolor{parText}{HTML}{000000}

\begin{tikzpicture}[
  >=Stealth,
  font=\footnotesize,
  updatebox/.style={
    draw=black!75,
    rounded corners=2pt,
    fill=parBox,
    line width=0.45pt,
    align=center,
    inner sep=2.5pt,
    minimum width=14mm,
    minimum height=8mm
  },
  combiner/.style={
    draw=black!75,
    rounded corners=2pt,
    fill=parBox,
    line width=0.45pt,
    align=center,
    font=\scriptsize,
    inner sep=2.5pt,
    minimum width=17mm,
    minimum height=8mm
  },
  flow/.style={->, line width=0.7pt, draw=black!75},
  data/.style={font=\scriptsize, align=center, text=parText}
]

\node[data] (prior) at (0,0)
  {\textcolor{black}{\(\mathcal F_k^-\)}};

\node[updatebox] (u1) at (2.4,0) {\(\mathcal U_{s_1}\)};
\node[updatebox] (u2) at (4.8,0) {\(\mathcal U_{s_2}\)};
\node[data] (dots) at (6.55,0) {\(\cdots\)};
\node[updatebox] (uM) at (8.3,0) {\(\mathcal U_{s_{|S_i|}}\)};
\coordinate (p1) at ($(u1.south west)!0.30!(u1.south east)$);
\coordinate (p2) at ($(u2.south west)!0.30!(u2.south east)$);
\coordinate (pM) at ($(uM.south west)!0.30!(uM.south east)$);
\coordinate (z1in) at ($(u1.south west)!0.70!(u1.south east)$);
\coordinate (z2in) at ($(u2.south west)!0.70!(u2.south east)$);
\coordinate (zMin) at ($(uM.south west)!0.70!(uM.south east)$);

\node[combiner] (prod) at (10.55,0)
  {Fused};
\node[data, anchor=west] (post) at (11.95,0)
  {\textcolor{black}{\(\mathcal F_k^+\)}};
\coordinate (prodTopOne) at ($(prod.north west)!0.25!(prod.north east)$);
\coordinate (prodTopTwo) at ($(prod.north west)!0.55!(prod.north east)$);
\coordinate (priorBusStart) at (0.95,-1.85);
\coordinate (priorBusEnd) at (pM |- priorBusStart);

\draw[line width=0.7pt, draw=black!75] (prior) -- (0.95,0);
\draw[line width=0.7pt, draw=black!75] (0.95,0) -- (priorBusStart);
\draw[line width=0.7pt, draw=black!75] (priorBusStart) -- (priorBusEnd);
\draw[flow] (p1 |- priorBusStart) -- (p1);
\draw[flow] (p2 |- priorBusStart) -- (p2);
\draw[flow] (pM |- priorBusStart) -- (pM);

\draw[flow] (u1.north) -- ++(0,0.50) -| (prodTopOne);
\draw[flow] (u2.north) -- ++(0,0.36) -| (prodTopTwo);
\draw[flow] (uM.east) -- node[data, above=0.4mm] {\(\gamma_{s_{|S_i|}}\)} (prod.west);

\node[data] at (2.15,0.78) {\(\gamma_{s_1}\)};
\node[data] at (4.55,0.66) {\(\gamma_{s_2}\)};

\draw[flow] (prod) -- (post);

\coordinate (zBase) at (0,-1.38);
\node[data] (z1) at (z1in |- zBase) {\(\V{Z}_{s_1,k}^{G}\)};
\node[data] (z2) at (z2in |- zBase) {\(\V{Z}_{s_2,k}^{G}\)};
\node[data] (zM) at (zMin |- zBase) {\(\V{Z}_{s_{|S_i|},k}^{G}\)};

\draw[flow] (z1.north) -- (z1in);
\draw[flow] (z2.north) -- (z2in);
\draw[flow] (zM.north) -- (zMin);

\end{tikzpicture}

%% file: sections/sec_handover.tex
\vspace{-1mm}
\section{Proposed Event-Triggered \grbp{} Handover}
\label{sec:handover}

Each \ac{BS} runs \revGrBP{GrBP as described in} \secref{sec:single_bs_method}. Unlike coordinated \grbpv{CS} and \grbpv{CP}, the proposed handover family exchanges information only between neighboring \acp{BS} sharing a track. Its unique \emph{owner} has authoritative track-management and handover responsibility. A \emph{shadow holder} retains a replicated \emph{shadow track}. Prior-only \grbpv{H} transfers the predicted track once when triggered. During subsequent co-observation frames, \grbpv{HM} additionally forwards associated measurement groups for measurement-level fusion, \grbpv{HL} forwards shadow-holder likelihood messages for product fusion, and \grbpv{HP} forwards full posterior beliefs for conservative fusion accounting for the shared prior. All three extensions require prior handover.

\begin{figure}[!t]
  \centering
  \resizebox{0.95\columnwidth}{!}{\input{Drawings/handover}}
  \caption{Proposed owner-centered \grbp{} handover processing: prior handover from the owner to shadow holders, shadow-to-owner measurement/likelihood/posterior packets, and owner-to-shadow fused-posterior broadcasts. Each experiment uses one of \grbpv{H}, \grbpv{HM}, \grbpv{HL}, or \grbpv{HP}. The evidence-exchange variants are not mixed.}
  \label{fig:handover-processing}
\end{figure}
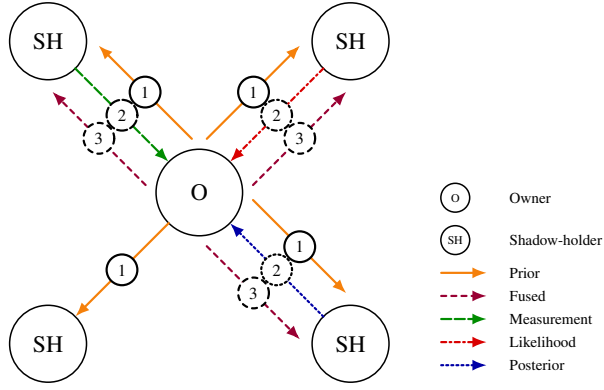

\vspace{-1mm}
\subsection{General Principles}
\label{sec:handover_general}

Per-frame time subscripts are omitted in this section. In the shorthand of
\secref{sec:single_bs_method}, the local update at \ac{BS}~\(s\) maps its
\rev{predicted-message collection \(\mathcal F_{s}^{-}\)} and grouped
measurements \(\V{Z}_{s}^{G}\) to a \rev{posterior-belief collection
\(\mathcal F_{s}^{+}\)}.
\figref{fig:handover-processing} illustrates the processing common to all four
variants: orange arrows denote prior packets, green links grouped-measurement
packets, red links likelihood packets, blue links posterior packets, and purple
links fused-posterior broadcasts.

\subsubsection{Track Replication and Shadow Holders}

Unlike conventional handover, a \ac{DISAC} network need not have the transmitting
\ac{BS} relinquish a target while it remains visible. A handover therefore
replicates rather than moves a track: the transmitter keeps the original while the
receiver (referred to as the \emph{shadow holder}) initializes a \emph{shadow track}. All four
variants share this and differ only in what is exchanged afterwards.
\subsubsection{Handover Score and Criterion}
\label{sec:handover_criteria}

All variants are gated by the same handover score, which measures the expected observability of a given track from the perspective of a neighboring \ac{BS}. For a predicted
component at the transmitting \ac{BS}~\(t\), let \(p_t\) denote its predicted
existence probability and \(f_t(\V{x},\M{E})\) its predicted density. The handover score
from the transmitting \ac{BS}~\(t\) to the receiving \ac{BS}~\(r\) is
\begin{align}
\Lambda_{t \to r}
&\triangleq{}
\ind\!\left\{p_t \ge p_{\mathrm{th}}\right\}
\int
\mu_{m,r}(\V{x},\M{E})\,
f_t(\V{x},\M{E})\,
\,\mathrm d\V{x}\,\mathrm d\M{E}. \label{eq:component_visibility_score}
\end{align}
Here, \(p_{\mathrm{th}}\) is the existence-probability threshold, and
\(\mu_{m,r}(\V{x},\M{E})\) is the mean number of target-generated detections at
the receiving \ac{BS}~\(r\), defined in~\eqref{eq:system_measurement_count_mean}.
The score \(\Lambda_{t\to r}\) is thus the expected number of detections the
component would generate at \ac{BS}~\(r\) once its existence probability clears the
gate \(p_{\mathrm{th}}\). It is zero for insufficiently confident tracks and large
only when the component also lies well inside the sensing region of \ac{BS}~\(r\). The
handover criterion for track \(\ell\) is
\begin{equation}
    \Lambda_{t \to r}(\ell) \ge \Lambda_{\mathrm{th}},
\label{eq:prior_handover_trigger}
\end{equation}
where \(\Lambda_{t \to r}(\ell)\) denotes \(\Lambda_{t \to r}\) evaluated for
the predicted existence probability and density of track \(\ell\), and
\(\Lambda_{\mathrm{th}}\) is a fixed threshold. Each variant uses this criterion as
its basic trigger, as detailed below.

\subsubsection{Local Label System}
\label{sec:ownership_transfer_prior_packets}

Handover requires explicit track management to prevent the same physical
target from being handed over or initialized more than once. Unlike the
multisensor \ac{MTT} literature, where label consistency is enforced through a
global label space, we use a fully local label system built on the current owner
\ac{BS} and that owner's local \ac{UID}.

\paragraph*{Owner--UID label} When a target is first established at a \ac{BS}, that
\ac{BS} becomes the owner and assigns a local \ac{UID}. The pair formed by the
current owner and that owner's \ac{UID} is the track label carried in prior,
measurement, likelihood, and posterior packets. Shadow holders store this
owner--UID label and use it to route optional information back to the owner.

\paragraph*{Ownership} Each track has exactly one owner \ac{BS}, which retains its
local \ac{UID} while it is predicted to observe the target. When the owner loses
visibility, ownership is transferred, via a two-message handshake over the
reliable, in-order backhaul of \secref{sec:system_model}, to the visible shadow
holder with the largest handover score \(\Lambda_{o\to r}\). The old owner keeps
authority until the acknowledgement arrives, so at most one owner is active at any
time, and the new owner \(o'\) is advertised through the updated owner--UID label
in subsequent packets. To preserve the neighbor-local structure, \(o'\) inherits
only its own sensing neighbors,
\(\mathcal{R}_{o'}(\ell)\leftarrow(\{o\}\cup\mathcal{R}_o(\ell))\cap\Set{B}_{o'}\),
and excluded holders are released until a fresh prior handover re-establishes them.
If no shadow holder is visible, the original owner is retained until the track is
recovered or pruned.

\paragraph*{Duplicate resolution} Because neighboring \acp{BS} may independently
initiate the same target before any prior is exchanged, an incoming prior is first
gated against the receiver's existing tracks: if it gates to an existing track, the
receiver merges the two (keeping the lexicographically smaller owner--UID pair as
the canonical label and notifying both owners by a bounded merge/alias
message) rather than appending a duplicate, so a single owner is preserved within
the neighbor-local handover cluster. As the merged priors may be correlated, the
canonical belief is kept or fused conservatively by \ac{GCI}, and a new shadow track
is created only when the prior gates to no existing track. 
\vspace{-1mm}
\subsection{Prior Handover (\grbpv{H})}
\label{sec:prior_handover}
For prior handover, the transmitting \ac{BS} is always the owner
of the track, so we denote it by \(o\) throughout this subsection.
\paragraph*{Trigger} A prior handover
from owner \(o\) to a neighbor \ac{BS}~\(r\) is triggered at the first frame in
which the criterion \eqref{eq:prior_handover_trigger} is met for that neighbor,
i.e., when the track is predicted to enter the sensing region of \ac{BS}~\(r\).
\begin{revision}A \emph{visibility episode} is a maximal run of consecutive
frames in which \(\Lambda_{o\to r}(\ell)\ge\Lambda_{\mathrm{th}}\). After the score
falls below the hysteresis release level, a later threshold upcrossing starts a
new episode and triggers a fresh prior handover.\end{revision}
It is event-triggered: one packet per (track, receiver) pair per visibility episode.

\paragraph*{Operation at the sending \ac{BS}} Only owners send priors. For each owned
track \(\ell\), the owner keeps a shadow-holder record
\(\mathcal{R}_o(\ell)\subseteq\Set{B}_o\) of the neighbors that already hold a
shadow track. Since each \ac{BS} knows its neighbors' sensing regions
(\secref{sec:system_model}), the owner evaluates the handover score
\(\Lambda_{o\to r}(\ell)\) of \eqref{eq:component_visibility_score} for every
neighbor \(r\in\Set{B}_o\) from its own predicted density, without feedback, and
treats \(r\) as observing \(\ell\) while \(\Lambda_{o\to r}(\ell)\ge\Lambda_{\mathrm{th}}\).
At each frame it forms the set of newly triggered receivers
\begin{equation}
\begin{aligned}
\mathcal{A}_{o}(\ell)
&\triangleq
\bigl\{
r\in \Set{B}_{o}\setminus\mathcal{R}_{o}(\ell):
\Lambda_{o\to r}(\ell)\ge\Lambda_{\mathrm{th}}
\bigr\},
\end{aligned}
\label{eq:shadow_holder_update}
\end{equation}
sends each \(r\in\mathcal{A}_o(\ell)\) a prior packet carrying the
{predicted message \(\alpha_o^\ell(\V{y}^{\ell})\)} and the owner--UID label (updated in place during an
ownership transfer), and refreshes the record\footnote{To avoid chattering near the
threshold, \(\Lambda_{\mathrm{th}}\) is applied with hysteresis, and a shadow holder
that prunes its local copy notifies the owner so that \(\mathcal{R}_o(\ell)\) does
not become stale.} to keep the still-observing holders and add the new ones,
\begin{equation}
\mathcal{R}_o(\ell)\ \leftarrow\
\bigl\{\,r\in\mathcal{R}_o(\ell):\Lambda_{o\to r}(\ell)\ge\Lambda_{\mathrm{th}}\,\bigr\}
\cup\mathcal{A}_o(\ell).
\label{eq:shadow_record_refresh}
\end{equation}

\paragraph*{Operation at the receiving \ac{BS}} A neighbor \(r\) initializes a new shadow
track from each received density and stores its label. Collecting the priors handed
to \(r\) this frame, the transferred predicted-message collection is
{\(\mathcal F^{-}_{o\to r}\triangleq
\{\alpha_o^\ell(\V{y}^{\ell}):\ell\ \text{owned by}\ o,\
r\in\mathcal{A}_o(\ell)\}\)}. The receiver appends these messages as additional
legacy targets, subject to the duplicate-resolution rule of
\secref{sec:ownership_transfer_prior_packets},
{\(\mathcal F_r^{-}\leftarrow\mathcal F_r^{-}\cup
\bigcup_{o\in\Set{B}_r}\mathcal F^{-}_{o\to r}\)},
before the local grouped-measurement update (orange arrows in
\figref{fig:handover-processing}), after which the single-\ac{BS} update proceeds
unchanged. The receiver needs nothing beyond the packets themselves.
\vspace{-1mm}
\subsection{Measurement Handover (\grbpv{HM})}
\label{sec:measurement_handover}

Prior-only handover exchanges information only at the triggering frame. As the
target then moves through the overlap region, no further information is shared.
Measurement handover adds per-frame measurement-level fusion among the \acp{BS} that
jointly observe a shared track \(\ell\), namely the owner \(o\) and its current
shadow holders, which form the co-observing set
$S_\ell\ \triangleq\ \{o\}\cup\mathcal{R}_o(\ell)$, 
i.e., the observing set \(S_i\) of \secref{sec:parallel_processing} specialized
to track \(\ell\). Fusion is organized as an \emph{owner-centered star}: the owner
is the per-track fusion hub, and each shadow holder communicates only with the
owner. Because every shadow holder received its prior over a sensing-neighbor edge
and the record refresh \eqref{eq:shadow_record_refresh} keeps
\(\mathcal{R}_o(\ell)\subseteq\Set{B}_o\), all owner--shadow links exist by
construction. 

\paragraph*{Trigger} Measurement handover is active in every frame in which a shadow
holder and the owner both still observe \(\ell\), i.e., while
\(r\in\mathcal{R}_o(\ell)\). Unlike prior handover, which fires once per visibility
episode, it recurs every frame.

\paragraph*{Operation at a shadow holder} After its local grouped-measurement update,
shadow holder \(r\in\mathcal{R}_o(\ell)\) identifies the measurement group most
strongly associated with \(\ell\),
\begin{equation}
j_r^{\star}(\ell)\ \triangleq\ \arg\max_{j}\ \Pr\!\bigl(a_r^{\ell}=j\bigr),
\label{eq:selected_group}
\end{equation}
where \(a_r^{\ell}\) assigns a measurement group, or the null outcome \(j=0\), to
track \(\ell\) at \ac{BS}~\(r\). The group \(\V{Z}^{G}_{r\to o}\triangleq\V{Z}_r^{\,j_r^{\star}(\ell)}\)
is sent to the owner if \(j_{\revGrammar{r}}^{\star}(\ell)\neq0\) (green links in
\figref{fig:handover-processing}). Otherwise, no measurement is forwarded.

\paragraph*{Operation at the owner} The owner collects the groups
\(\{\V{Z}^{G}_{r\to o}\}_{r\in\mathcal{R}_o(\ell)}\) and fuses them with its own
measurements at the measurement level, applying the grouped-measurement likelihood
of each originating \ac{BS} sequentially, as in the sequential coordinated
processing of \secref{sec:sequential_processing} but restricted to \(S_\ell\). Evaluating
 a received group requires the originating \ac{BS}'s measurement model,
which follows from the shared geometry. The owner then broadcasts the fused
posterior to the shadow holders, which adopt it as their local posterior for
\(\ell\). 
\vspace{-1mm}
\subsection{Likelihood-Message Handover (\grbpv{HL})}
\label{sec:likelihood_handover}

Likelihood-message handover is the track-level counterpart of measurement
handover: it uses the same owner-centered star over \(S_\ell\) and the same
per-frame trigger, but the shadow holders send {likelihood messages rather
than raw measurement groups}.

\paragraph*{Operation at a shadow holder} After its local update, shadow holder
\(u\in\mathcal{R}_o(\ell)\) sends the owner its likelihood message
{\(\gamma_u^\ell(\V{y})\propto
\Tilde{f}_u^\ell(\V{y})/\alpha^\ell(\V{y})\)}, i.e., the
\ac{BP} output of \secref{sec:single_bs_method} that summarizes its current-frame
evidence (red links in \figref{fig:handover-processing}). 

\paragraph*{Operation at the owner} The owner combines its own message with the
received ones using the parallel-processing rule of \secref{sec:parallel_processing},
retaining a single copy of the prior,
\begin{equation}
{\Tilde{f}^{\ell}(\V{y})\ \propto\ \alpha^{\ell}(\V{y})
\!\!\prod_{s\in S_\ell}\!\!\gamma_s^{\ell}(\V{y})},
\label{eq:posterior_fused}
\end{equation}
and broadcasts the fused posterior \(\Tilde{f}^{\ell}(\V{y})\) to the shadow
holders, which adopt it. Because all members of \(S_\ell\) then store the same
posterior and predict it with the same transition model, they share a common prior
at the next frame, so the prior divided out in \(\gamma_s\) is common by
construction. Equation~\eqref{eq:posterior_fused} is then exact within the
common-prior, fixed-track, factorized-message model of
\secref{sec:parallel_processing}. The underlying \ac{BP} messages remain
approximate, so it does not equal the centralized multi-object posterior.
\subsection{Posterior Handover (\grbpv{HP})}
\label{sec:posterior_handover}

Posterior handover uses the same owner-centered star and per-frame trigger as
likelihood-message handover, but each shadow holder forwards its full posterior
belief \(\Tilde{f}_r(\V{y}^{\ell})\) rather than the likelihood message
\(\gamma_r\). The owner then fuses the received posteriors with its own. Because
the forwarded posteriors share the common prior \(f(\V{y}^{\ell})\), multiplying
them as in \eqref{eq:posterior_fused} would multiply-count that prior. The owner
therefore fuses them conservatively by equal-weight \ac{GCI}
\cite{garcia2025distributed},
\begin{equation}
\Tilde{f}(\V{y}^{\ell})\ \propto\ \prod_{s\in S_\ell}\Tilde{f}_s(\V{y}^{\ell})^{1/|S_\ell|},
\label{eq:posterior_handover_gci}
\end{equation}
and broadcasts the fused posterior to the shadow holders, which adopt it.

\subsection{Scalability Analysis}
\label{sec:handover_scalability}

We assess the proposed \grbp{} handover variants \grbpv{H}, \grbpv{HM}, \grbpv{HL}, and \grbpv{HP} against Definition~\ref{def:scalable}. Recall the bounded sensing-graph degree
\(d_{\max}\) (Assumption~\ref{as:degree}), and let \(n^{t}_{s,k}\) and
\(n^{g}_{s,k}\) be the realized numbers of active legacy targets and local
measurement groups at \ac{BS}~\(s\), whose second moments are bounded by
Assumption~\ref{as:workload}. The handover records are local: the owner's
shadow-holder record \(\mathcal{R}_o(\ell)\) is a subset of its \(\le d_{\max}\)
sensing neighbors, so each shared track has at most \(d_{\max}\) active
owner--shadow edges.

We write \(b_{\mathrm{evt}}\) for a finite upper bound on the total event- and
control-byte load per track--neighbor pair per frame (i.e., the prior packet
together with any
ownership-transfer, acknowledgement, merge/alias, or release messages, which may
comprise several packets) and \(b_Z\), \(b_\gamma\), \(b_Q\) for finite upper bounds
on the bytes of one measurement group, one likelihood message, and one
fused-posterior packet, respectively.

\begin{proposition}[Scalability of the proposed handover]
\label{prop:handover_scalable}
Fix a \ac{BS}~\(s\) and frame \(k\). With finite \(b_{\mathrm{evt}},b_Z,b_\gamma,b_Q\),
the realized per-node loads of all four proposed \grbp{} handover variants satisfy
\begin{align}
C^{\mathrm{comp}}_{s,k}
&= O\!\big(n^{t}_{s,k}n^{g}_{s,k} + d_{\max}n^{t}_{s,k}\big),
\label{eq:handover_comp_bound}\\
C^{\mathrm{comm}}_{s,k}
&\le
\underbrace{2 d_{\max}n^{t}_{s,k} b_{\mathrm{evt}}}_{\text{event-triggered}}
+
\underbrace{d_{\max}n^{t}_{s,k}\,\beta}_{\text{optional, per active edge}},
\label{eq:handover_comm_bound}
\end{align}
with \(\beta=0\) for prior only, \(\beta=b_Z+b_Q\) for measurement,
\(\beta=b_\gamma+b_Q\) for likelihood-message, and \(\beta=2b_Q\) for posterior
handover, and hidden constants fixed by
Assumption~\ref{as:representation}. These loads depend only on the local counts and
the neighborhood, not on \(N\). Taking expectations and applying Cauchy--Schwarz,
\(\mathbb E[n^{t}_{s,k}n^{g}_{s,k}]\le(\mathbb E[(n^{t}_{s,k})^2]\,\mathbb E[(n^{g}_{s,k})^2])^{1/2}\),
Assumptions~\ref{as:degree}--\ref{as:representation} bound the expected per-frame
computation and the long-run per-node communication by \(N\)-independent constants.
Hence \grbpv{H}, \grbpv{HM}, \grbpv{HL}, and \grbpv{HP} are all scalable in the sense of Definition~\ref{def:scalable}.
\end{proposition}
\begin{proof}
    See Appendix~\ref{sec:app_handover_scalability_proofs}.
\end{proof}

The two terms in \eqref{eq:handover_comm_bound} separate event-triggered from per-frame traffic. Prior handover
contributes only the first term: a prior packet is sent to a neighbor at most once
per visibility episode (a threshold upcrossing of
\eqref{eq:prior_handover_trigger}) so on average this term is well below its
per-frame cap and is independent of how long targets dwell in overlaps. The
measurement, likelihood-message, and posterior variants add the second term,
incurred on every active owner--shadow edge, which therefore grows with the
overlap dwell time. All four variants are
scalable under the same conditions (Assumptions~\ref{as:degree}--\ref{as:representation}). Prior
handover has the smallest 
communication constant, while the three extensions add a
per-frame term of size \(b_Z+b_Q\), \(b_\gamma+b_Q\), and \(2b_Q\) per active edge,
respectively.



%% file: Drawings/handover.tex
\usetikzlibrary{calc,arrows.meta,matrix}

\begin{tikzpicture}[
  >=Latex,
  shnode/.style={
    circle,
    draw=black,
    line width=0.9pt,
    minimum size=16mm,
    inner sep=0pt,
    font=\Large,
    text=black
  },
  onode/.style={
    circle,
    draw=black,
    line width=0.9pt,
    minimum size=18mm,
    inner sep=0pt,
    font=\Large,
    text=black
  },
  arrow/.style={
    -{Latex[length=3mm,width=2.3mm]},
    line width=1.3pt,
    line cap=round
  },
  prior/.style={
    arrow,
    draw=orange!85!brown
  },
  fused/.style={
    arrow,
    draw=purple!80!black,
    dashed
  },
  meas/.style={
    arrow,
    draw=green!55!black,
    dash pattern=on 6pt off 2pt
  },
  like/.style={
    arrow,
    draw=red!85!black,
    dash pattern=on 3pt off 2pt on 1pt off 2pt
  },
  post/.style={
    arrow,
    draw=blue!65!black,
    dash pattern=on 1.2pt off 1.8pt
  },
  circnum/.style={
    circle,
    draw=black,
    fill=white,
    minimum size=6.4mm,
    inner sep=0pt,
    font=\normalsize,
    midway
  }
]

\def\R{4.45cm}      
\def\rO{9mm}        
\def\rSH{8mm}       
\def\sep{6.8mm}     

\node[onode]  (O)    at (0,0)       {O};

\node[shnode] (SH1)  at (135:\R)    {SH};   
\node[shnode] (SH2)  at (45:\R)     {SH};   
\node[shnode] (SH3)  at (-45:\R)    {SH};   
\node[shnode] (SH4)  at (-135:\R)   {SH};   

%
\newcommand{\tripleedge}[6]{%
  \coordinate (a) at ($(#1)!#3!(#2)$);%
  \coordinate (b) at ($(#2)!#4!(#1)$);%
  \ifnum#5=1
    \coordinate (pstart) at ($(a)!\sep!90:(b)$);
    \coordinate (pend)   at ($(b)!\sep!-90:(a)$);
    \coordinate (mstart) at (b);
    \coordinate (mend)   at (a);
    \coordinate (fstart) at ($(a)!\sep!-90:(b)$);
    \coordinate (fend)   at ($(b)!\sep!90:(a)$);
  \else
    \coordinate (pstart) at ($(a)!\sep!-90:(b)$);
    \coordinate (pend)   at ($(b)!\sep!90:(a)$);
    \coordinate (mstart) at (b);
    \coordinate (mend)   at (a);
    \coordinate (fstart) at ($(a)!\sep!90:(b)$);
    \coordinate (fend)   at ($(b)!\sep!-90:(a)$);
  \fi
  \draw[prior] (pstart) -- node[circnum] {1} (pend);
  \draw[#6]    (mstart) -- node[circnum] {2} (mend);
  \draw[fused] (fstart) -- node[circnum] {3} (fend);
}

%
\newcommand{\singleprioredge}[4]{%
  \coordinate (a) at ($(#1)!#3!(#2)$);%
  \coordinate (b) at ($(#2)!#4!(#1)$);%
  \draw[prior] (a) -- node[circnum] {1} (b);
}

\tripleedge{O}{SH1}{\rO}{\rSH}{-1}{meas}
\tripleedge{O}{SH2}{\rO}{\rSH}{1}{like}
\tripleedge{O}{SH3}{\rO}{\rSH}{1}{post}
\singleprioredge{O}{SH4}{\rO}{\rSH}

\newcommand{\legendline}[1]{%
  \tikz[baseline=-0.6ex]\draw[#1] (0,0) -- (0.95,0);%
}

\newcommand{\legendowner}{%
  \tikz[baseline=-0.6ex]\node[onode,minimum size=6mm,font=\scriptsize]{O};%
}

\newcommand{\legendsh}{%
  \tikz[baseline=-0.6ex]\node[shnode,minimum size=6mm,font=\scriptsize]{SH};%
}

\matrix[
  matrix of nodes,
  anchor=west,
  at={(4.85,-1.75)},
  row sep=1.8mm,
  column sep=3.8mm,
  nodes={anchor=west,font=\small,inner sep=1pt}
]{
  \legendowner       & Owner          \\
  \legendsh          & Shadow-holder  \\
  \legendline{prior} & Prior          \\
  \legendline{fused} & Fused          \\
  \legendline{meas}  & Measurement    \\
  \legendline{like}  & Likelihood     \\
  \legendline{post}  & Posterior      \\
};

\end{tikzpicture}

%% file: sections/sec_numerical_results.tex
\vspace{-1mm}
\section{Numerical Results}
\label{sec:numerical_results}
\revGrBP{We evaluate GrBP deployments in a seven-\ac{BS} \ac{DISAC} network, then compare uncoordinated distributed GrBP with the state-of-the-art BP-based ETT method.} Detailed parameters and per-\ac{BS} results are in the repository cited in \secref{sec:introduction}.
\vspace{-1mm}

\subsection{Simulation Setup}

\subsubsection{Compared Methods and Labels}
\label{subsec:compared_methods}
\revGrBP{GrBP denotes the common local \ac{MTT} method introduced in \secref{sec:introduction}. Suffixes identify the processing architecture, handover payload, or oracle-grouping condition.}

\revGrBP{The suffix O denotes oracle grouping using known measurement-source identities, not noiseless measurements or a different local method. All non-oracle GrBP variants use the clustering procedure below.}

\paragraph*{Coordinated \grbp{} references}
\grbpv{CS} and \grbpv{CP} use the sequential and parallel schedules of \secref{sec:coordinated_processing}. \grbpv{CO} uses \grbpv{CS} with oracle grouping.

\paragraph*{Uncoordinated distributed \grbp{} baseline}
\grbpv{D} runs \revGrBP{GrBP} independently at each \ac{BS}, without inter-\ac{BS} exchange. \grbpv{DO} adds oracle grouping.

\paragraph*{Proposed \grbp{} handover variants}
\grbpv{H} is the proposed event-triggered prior-only handover. \grbpv{HM}, \grbpv{HL}, and \grbpv{HP} add measurement, likelihood-message, and posterior exchange, respectively (\secref{sec:handover}). \grbpv{HMO} uses \grbpv{HM} with oracle grouping, not a different protocol.

\paragraph*{\revGrBP{State-of-the-Art BP-Based ETT Method}}
\label{subsec:ett_reference}
\revGrBP{The state-of-the-art \ac{BP}-based ETT method of \cite{meyer2021scalable}} runs independently at each \ac{BS}, without inter-\ac{BS} exchange. \revGrBP{ETT denotes this reference method, not the tracking task or a \grbp{} variant.} It associates individual measurements without \revGrBP{the GrBP grouping step}. We use the authors' reference code without its preprocessing stage\footnote{Reference implementation: \url{https://github.com/meyer-ucsd/EOT-TSP-21}.} and adapt its measurement-count calculation to \(\mu_{m,s}(\V{x},\M{E})\) in \eqref{eq:system_measurement_count_mean}, with the same resolution- and visibility-aware parameters as \revGrBP{GrBP}. This avoids a mismatched target-generated count mean. The implementations and timing boundaries still differ.

\subsubsection{Network Geometry}
The simulation geometry is shown in \figref{fig:simulation-scenarios}. The \ac{DISAC} network has one \ac{BS} at $(0,0,15)$\,\revYu{m} and the remaining \acp{BS} at \((\pm150,0,15)\)\,\revYu{m}, \((\pm75,130,15)\)\,\revYu{m}, and \((\pm75,-130,15)\)\,\revYu{m}. The backhaul graph connects \ac{BS}~0 to all outer \acp{BS}, and each outer \ac{BS} to \ac{BS}~0 and two neighboring outer \acp{BS}. Each base station has \revGrammar{a} sensing radius \(\gamma = 120 \mathrm{m}\).

\begin{figure}[!t]
  \centering
  \resizebox{0.9\linewidth}{!}{\input{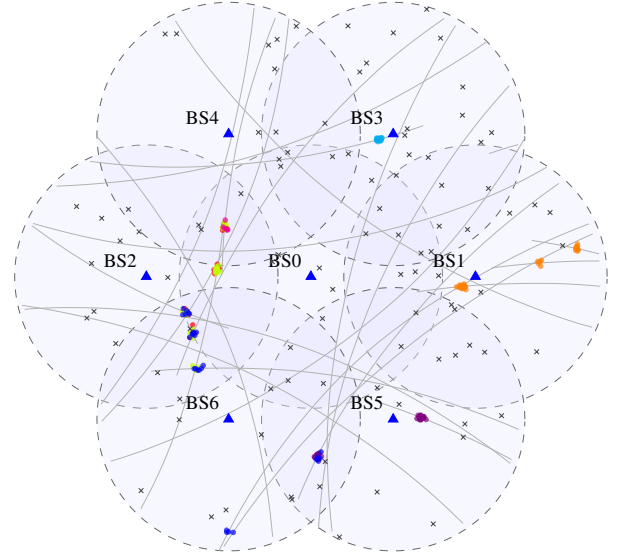}}
  \caption{Simulation scenario for the \ac{DISAC} network used in the numerical study. Light-blue discs indicate the \acp{FOV}, target-generated measurements are color-coded by \ac{BS}, and black crosses indicate clutter.}
  \label{fig:simulation-scenarios}
\end{figure}

\subsubsection{Generative Model} The detailed generative model can be accessed via the provided URL. Targets move in a 2D plane with a \ac{CTRV} motion model. The state is \(\V{x}_k^i=[x_k^i,y_k^i,\theta_k^i,v_k^i,\omega_k^i]^\T\), where \(x,y\) \revGrammar{are} the 2D position \revGrammar{coordinates} and \(\theta\) is the heading. The interval \revGrammar{is} \(\Delta t=1\)~s. The nominal speed is $v=7.0$\,m/s, the nominal turn rate is $\omega=0.01047$\,rad/s, and the generator adds per-scan perturbations with \(\sigma_v=0.1\)~m/s and \(\sigma_{\omega}=0.005\)~rad/s. Six targets enter from the boundary at the first scan, and one more boundary target is spawned every $5$ scans over a $100$-scan sequence. The initial target positions are uniformly distributed over designated regions.

All true targets use an elliptical extent with principal \revGrammar{axes} of $5.0$ and $2.0$, rotated with the target heading. Measurements are range--bearing with standard \revGrammar{deviations} \(\sigma_r=0.5\)~m and \(\sigma_{\theta}=0.5^\circ\)\revGrammar{,} respectively. The target-generated detection count follows the resolution- and visibility-aware mean in \eqref{eq:system_measurement_count_mean}, using \(B=100\)~MHz, \(\lambda_{\mathrm c}=0.01\)~m, \(D_{\mathrm T}=D_{\mathrm R}=0.5\)~m, \(p_{\mathrm d}^{\mathrm{cell}}=0.9\), \(\kappa=1\), and minimum target-to-\ac{BS} range \(r_{\min}=30\)~m. Clutter is a \ac{PPP} with mean \(\mu_{\mathrm{fa}}=15\) detections per \ac{BS} per scan, uniformly distributed over the local \ac{FOV}.

\subsubsection{\revGrBP{GrBP Parameters and Measurement Grouping}}

\revGrBP{GrBP} uses survival probability \(p_s=0.95\), newborn mean \(\mu_n=1\), newborn extent prior \(\mathrm{diag}(9,9)\), and inverse-Wishart initialization \revGrammar{degrees} of freedom \(\nu_0=10\). The birth location of a new target is heuristically sampled around the measurement in the same way as \revGrammar{in} \cite{meyer2021scalable}. The initial heading, velocity, and turn rate are all set to \(0\). The kinematic birth covariance is \(\operatorname{diag}(P_{\mathrm{init}},P_{\mathrm{init}},\sigma_{\theta,0}^2,\sigma_{v,0}^2,\sigma_{\omega,0}^2)\), with \(P_{\mathrm{init}}=100~\mathrm{m}^2\), \(\sigma_{\theta,0}^2=\pi^2~\mathrm{rad}^2\), \(\sigma_{v,0}^2=100~\mathrm{m}^2/\mathrm{s}^2\), and \(\sigma_{\omega,0}^2=0.25~\mathrm{rad}^2/\mathrm{s}^2\). Targets are pruned below \revGrammar{an} existence \revGrammar{probability of} \(10^{-3}\), and the detection threshold is set to $0.9$.
All non-oracle \grbp{} variants retain one deterministic maximum-likelihood \ac{DBSCAN} partition per scan. Clustering forms the groups. \ac{BP} associates whole groups to legacy or newborn targets. The grouped association problem is solved by loopy \ac{BP} until the maximum message change is below \(10^{-5}\) or $100$ iterations are reached. The handover criterion is set at \(\Lambda_{\mathrm{th}}=2\) for the proposed prior-only variant \grbpv{H}.\footnote{We swept the handover visibility threshold \(\Lambda_{\mathrm{th}}\in[0.5,5]\) per \ac{BS}. No single threshold dominates all metrics, and the threshold should thus be interpreted as an accuracy-versus-conservatism tuning parameter rather than a universal optimum.} Because clutter-only clusters are nearly singleton in this setup, with mean size \(1.02\), we use the approximation \(\mu_g\approx\mu_c\).

\begin{figure}
  \centering
  \resizebox{0.90\linewidth}{!}{\input{Drawings/measurement_model_validation.tikz}}
  \caption{Validation of the measurement-count mean under full visibility. The
  black curve shows the model prediction and the orange curve shows the mean
  and standard deviation of \revGrBP{\acf{OMP}} detections over \(200\) Monte Carlo realizations per distance value.}
  \label{fig:measurement-model-validation}
\end{figure}
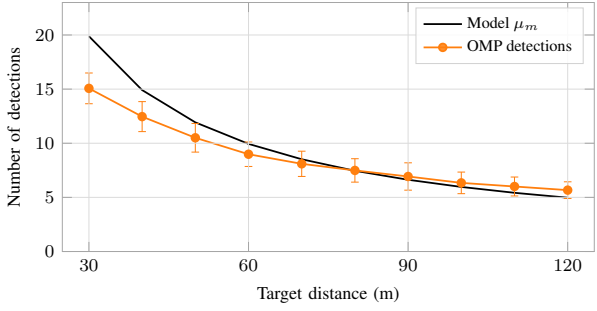

\vspace{-1mm}
\subsection{Simulation Study Results}

We first validate the measurement-count model, then compare coordinated and uncoordinated \grbp{} baselines with the proposed handover variants. Next, we isolate grouping error and assess per-node scaling and communication. Finally, \revGrBP{we compare \grbpv{D} and the BP-based ETT method in accuracy and runtime under uncoordinated distributed deployment.}


\begin{table*}[!t]
  \centering
  \caption{Tracking performance over \(100\) trials: mean $\pm$ standard deviation across \acp{BS}, \revGrBP{for the} \grbp{} \revGrBP{deployments and the BP-based ETT method. O denotes oracle grouping.}}
  \label{tab:per-bs-results}
  \begingroup
  \footnotesize
  \setlength{\tabcolsep}{3pt}
  \renewcommand{\arraystretch}{1.05}
  \input{result/per_bs_results_summary.tex}
  \endgroup
\end{table*}

\subsubsection{Validation of the Proposed Measurement Model}
\label{subsec:measurement_model_validation}

We test the measurement-count model of \secref{sec:detection_rate} under full visibility with the simulation parameters above. At each range, dense sub-resolution scatterers fill the continuous footprint \(\mathcal{A}_\kappa(\V{x},\M{E})\), and a noiseless \ac{MIMO}-\ac{OFDM} channel is synthesized. An \ac{OMP} estimator recovers the resolvable detection count for comparison with \eqref{eq:system_measurement_count_mean} using \eqref{eq:rho_from_pd_cell_density}. In \figref{fig:measurement-model-validation}, the model overpredicts counts at short range and underpredicts them near the longest ranges, but reproduces the dominant monotonic decrease under full visibility.

\subsubsection{\grbp{} Baselines and Proposed Handover Results}

\figref{fig:all-bs-pm-summary} averages \ac{GOSPA} over seven \acp{BS} and \(100\) trials. Coordinated \grbpv{CS} and \grbpv{CP} outperform uncoordinated \grbpv{D}, and oracle-grouped \grbpv{CO} further improves accuracy. \grbpv{D} has the largest errors for most frames, consistent with target reinitialization at \ac{FOV} entry. Larger increases occur near changes in the local true-target count (black dashed curve), but these aggregate results do not isolate boundary events or establish causality.

The proposed handover variants narrow the gap to coordinated \revGrBP{GrBP processing}. Among non-oracle variants, \grbpv{HP} has the lowest aggregate \ac{GOSPA}, followed by \grbpv{HL}, while \grbpv{H} and \grbpv{HM} remain closer to \grbpv{D}. Oracle-grouped \grbpv{HMO} further reduces error. \revGrBP{The BP-based ETT comparison follows below.}
\begin{figure}[!t]
  \centering
  \resizebox{0.96\linewidth}{!}{\input{Drawings/all_bs_pm_gospa_timeseries.tikz}}
  \vspace{0.2em}

  \begingroup
    \DocumentLegendFont
    \newcommand{\AllBsPmLegendItem}[2]{\tikz[baseline=-0.55ex]{\draw[#1] (0,0) -- (1.35em,0);}~#2}
    Coordinated \grbp{}:~
    \AllBsPmLegendItem{thick, color=plotC}{\grbpv{CS}}\hspace{0.42em}
    \AllBsPmLegendItem{thick, color=plotCL}{\grbpv{CP}}\hspace{0.42em}
    \AllBsPmLegendItem{thick, dashed, color=plotCO}{\grbpv{CO}}\\[-0.1ex]
    Distributed \grbp{}:~
    \AllBsPmLegendItem{thick, color=plotD}{\grbpv{D}}\hspace{0.42em}
    \AllBsPmLegendItem{thick, dashed, color=plotDO}{\grbpv{DO}}\\[-0.1ex]
    Proposed \grbp{} handover:~
    \AllBsPmLegendItem{thick, color=plotHM}{\grbpv{H}}\hspace{0.42em}
    \AllBsPmLegendItem{thick, color=plotHWM}{\grbpv{HM}}\hspace{0.42em}
    \AllBsPmLegendItem{thick, dashed, color=plotHMO}{\grbpv{HMO}}\\[-0.1ex]
    \AllBsPmLegendItem{thick, color=plotHL}{\grbpv{HL}}\hspace{0.42em}
    \AllBsPmLegendItem{thick, color=plotHP}{\grbpv{HP}}\\[-0.1ex]
    Reference:~\AllBsPmLegendItem{thick, densely dotted, color=plotETT}{\revGrBP{ETT (BP, distributed)}}\hspace{0.42em}
    \AllBsPmLegendItem{thick, dashed, black}{Avg. track count}
  \endgroup
  \caption{GOSPA averaged over all \acp{BS} and \(100\) trials for \(100\) frames. The black dashed curve gives the average true-target count in local surveillance regions. Per-\ac{BS} series are in the repository.}
  \label{fig:all-bs-pm-summary}
\end{figure}
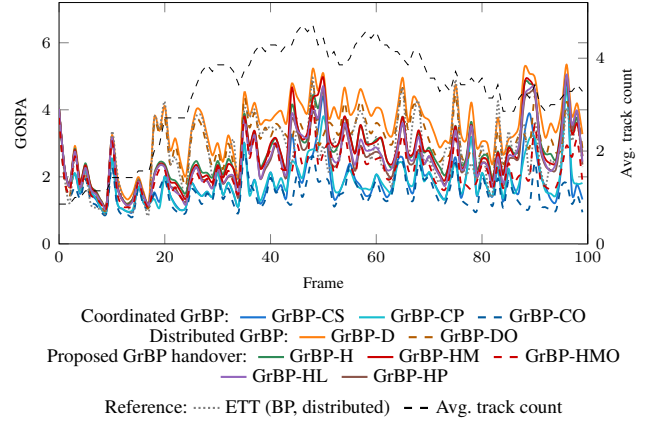
In \tabref{tab:per-bs-results}, \grbpv{CS}, \grbpv{CP}, and \grbpv{CO} achieve \ac{GOSPA} values of \(2.02\), \(1.99\), and \(1.61\), versus \(3.50\) for \grbpv{D}, which also has the most missed targets (\(13.61\) per \ac{BS}-trial). All proposed handover variants improve on \grbpv{D}. Among non-oracle handover variants, \grbpv{HP} has the lowest \ac{GOSPA} (\(2.59\)), missed-target count (\(7.60\)), and false-target count (\(1.18\)). It shares the lowest displayed localization component (\(1.55\)~m) with \grbpv{HL}. Oracle-grouped \grbpv{HMO} achieves \ac{GOSPA} \(2.24\).

\subsubsection{Clustering Error Contribution}
To isolate grouping error, we replace estimated \ac{DBSCAN} groups with oracle groups in matched pairs \grbpv{CS}--\grbpv{CO}, \grbpv{D}--\grbpv{DO}, and \grbpv{HM}--\grbpv{HMO}, retaining each \grbp{} architecture. Unlike \revGrBP{the ETT method's} individual-measurement association, these comparisons retain group-level association and one partition.

Across seven \acp{BS}, \(8.63\%\) of target detections are grouped with clutter and \(1.80\%\) of clutter detections with targets. Mean true-target and algorithm group sizes are similar (\(8.42\) versus \(8.50\)), but this does not imply correct membership. Clutter-only clusters are nearly singleton (\(1.02\)). The true-target cluster split rate is \(1.14\%\). Oracle grouping reduces \ac{GOSPA} from \(2.02\) to \(1.61\), \(3.50\) to \(3.03\), and \(2.79\) to \(2.24\) for the three pairs, respectively. These \(13\)--\(20\%\) improvements quantify grouping error, not changes of protocol or \revGrBP{local method}.

\subsubsection{Scalability Analysis}
Under the stated assumptions, \grbpv{D} and the proposed \grbpv{H}, \grbpv{HM}, \grbpv{HL}, and \grbpv{HP} have bounded expected per-node loads, unlike centralized coordinated \revGrBP{GrBP processing} (Propositions~\ref{prop:centralized_nonscalable}--\ref{prop:handover_scalable}). \figref{fig:scalability-vs-N} illustrates this network-size scaling, distinct from implementation runtime below, using a tiled hexagonal deployment with fixed target density and bounded neighborhoods. Maximum per-node loads follow the fixed payload-byte and runtime accounting model.\footnote{The communication accounting uses payload bytes only and excludes protocol headers. With \(b=8\) bytes per scalar, measurement dimension \(d_m=2\), kinematic dimension \(d_x=5\), and extent-scale dimension \(d_E=3\), one range--bearing measurement packet has \(b_z=d_m b=16\) bytes, and one prior, likelihood/track-level, posterior, or fused-posterior packet has \(b_\pi=b_\gamma=b_Q=[d_x+d_x^2/2+d_E+4]b=196\) bytes. For \grbpv{HM}, the prior-handover byte count should add \(b_z\) times the maximum forwarded-measurement count. The subscript on each runtime proxy identifies the corresponding \grbp{} variant. For example, $T_D$ is the proxy for \grbpv{D}. The runtime proxy uses per-measurement cost \(\tau_z=10^{-4}\)~s and per-track cost \(\tau_x=10^{-3}\)~s. Let \(m_{\max}\) be the maximum frame measurement count and \(c_{\max}\) the maximum number of tracks crossing to another \ac{BS} field of view. With fixed target density \(3\) per \ac{BS}, the plotted computation values are \(T_{CS}=m_{\max}\tau_z+3N\tau_x\), \(T_D=(m_{\max}/N)\tau_z+3\tau_x\), \(T_H=T_D+c_{\max}\tau_x\), \(T_{HM}=T_H+c_{\max}\tau_x\), and \(T_{HL}=T_{HP}=T_H+c_{\max}\tau_z\), converted to ms/frame.} The fusion-node load grows linearly in $N$, whereas the proposed handover variants remain bounded. Prior-only \grbpv{H} has the smallest constant. \grbpv{HM}, \grbpv{HL}, and \grbpv{HP} add evidence-exchange costs. With no inter-\ac{BS} communication, \grbpv{D} appears only as the local-computation reference.
\begin{figure}[!t]
  \centering
  \subfloat[Communication\label{fig:scalability-comm}]{%
    \begin{minipage}[t]{0.48\columnwidth}
      \centering
      \resizebox{\linewidth}{!}{\input{Drawings/scalability_comm_vs_N.tikz}}%
    \end{minipage}%
  }\hfill
  \subfloat[Computation\label{fig:scalability-comp}]{%
    \begin{minipage}[t]{0.48\columnwidth}
      \centering
      \resizebox{\linewidth}{!}{\input{Drawings/scalability_comp_vs_N.tikz}}%
    \end{minipage}%
  }
  \caption{Maximum per-node (a) communication and (b) computation load versus the number of \acp{BS} $N$ at fixed target density. The \grbpv{CS} fusion-node load grows linearly in $N$, while the proposed \grbp{} handover variants remain bounded under the stated assumptions.}
  \label{fig:scalability-vs-N}
\end{figure}
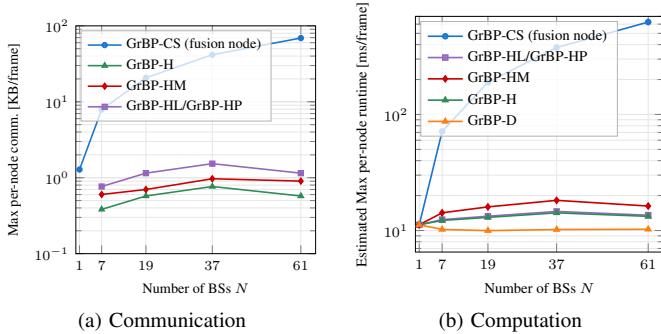

\tabref{tab:handover-counting} applies the same payload accounting to the seven-\ac{BS} network.

\begin{table}[!t]
  \centering
  \caption{Average per-trial transferred quantities and communication load.}
  \label{tab:handover-counting}
  \scriptsize
  \setlength{\tabcolsep}{3pt}
  \renewcommand{\arraystretch}{1.0}
  \input{result/handover_counting_summary.tex}
\end{table}

Coordinated \revGrBP{GrBP processing} requires \(468\)~KB per trial, versus \(32\), \(161\), \(344\), and \(343\)~KB for \grbpv{H}, \grbpv{HM}, \grbpv{HP}, and \grbpv{HL}. Thus, \grbpv{H} reduces localization RMSE by \(17.97\%\) relative to \grbpv{D} using \(6.9\%\) of the coordinated load. Evidence-exchange variants add shadow-to-owner packets and accepted-message fused-posterior returns during eligible co-observation frames. \grbpv{HM} remains cheaper because it forwards 16-byte measurements rather than the 196-byte track-level packets of \grbpv{HP} and \grbpv{HL}.

\subsubsection{\revGrBP{Comparison With State-of-the-Art BP-Based ETT}}
\label{subsec:pm_ett_comparison}
Both runtime panels compare \grbpv{D}, with estimated \ac{DBSCAN} grouping, against \revGrBP{the BP-based ETT method run independently at each BS}. Neither uses handover or oracle grouping. \revGrBP{The timers therefore exclude all GrBP handover variants and coordinated schedules.} ETT improves \ac{GOSPA} from \(3.50\) to \(2.86\) and missed-target count from \(13.61\) to \(12.29\) per \ac{BS}-trial (\tabref{tab:per-bs-results}). \revGrBP{Lower \grbpv{D} runtime does not imply higher accuracy.}

In \figref{fig:runtime-ett-vs-pm}(a), each point averages 700 \ac{BS}-frame samples (100 trials, seven \acp{BS}). Over 100 frames, ETT's mean runtime is 29.8--76.7 times \grbpv{D}'s. Overall means are 4156.67 and 64.83~ms/\ac{BS}-frame, respectively, a factor of 64.12. These are implementation timers, not end-to-end latency or equal-backend kernel timings. \grbpv{D}'s NumPy/SciPy CPU timer includes measurement preparation and clustering. \revGrBP{The ETT implementation} uses PyTorch CUDA and CPU association, timed after measurement extraction and conversion.

Panel~(b) uses the immediately preceding, otherwise matched 100-trial production run, averaging seven \acp{BS}. Mean BP-message times are 1.45~ms/\ac{BS}-frame for \grbpv{D} and 7.66 for ETT, a factor of 5.29. Both timers cover CPU-side operations and exclude likelihood construction, device transfers, and posterior updates. However, \grbpv{D} times the complete BP-association stage, whereas \revGrBP{the ETT implementation} times the inner per-measurement sum-product update: the timed kernels are not identical.

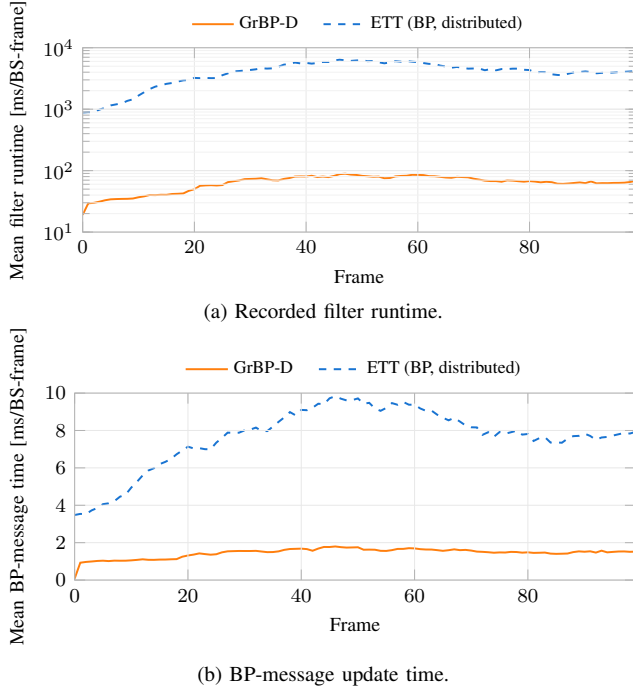
\begin{figure}[!t]
  \centering
  \subfloat[Recorded filter runtime.]{%
    \resizebox{0.96\linewidth}{!}{\input{Drawings/runtime_comparison_eot_vs_pm.tikz}}%
  }\\[-0.5ex]
  \subfloat[BP-message update time.]{%
    \resizebox{0.96\linewidth}{!}{\input{Drawings/bp_message_runtime_eot_vs_pm.tikz}}%
  }
  \caption{Distributed \grbpv{D} (solid orange) versus \revGrBP{the state-of-the-art BP-based ETT method} (dashed blue), without handover or oracle grouping, averaged over 100 trials and seven \acp{BS}. (a) Filter runtime, log scale. (b) BP-message time, linear scale, from archived timers of the preceding matched run.}
  \label{fig:runtime-ett-vs-pm}

\end{figure}

%% file: Drawings/measurement_model_validation.tikz
\begin{tikzpicture}
  \begin{axis}[
    width=0.94\columnwidth,
    height=0.44\columnwidth,
    scale only axis,
    xmin=25, xmax=125,
    ymin=0, ymax=23,
    xtick={30,60,90,120},
    ytick={0,5,10,15,20},
    xlabel={Target distance (m)},
    ylabel={Number of detections},
    label style={font=\footnotesize},
    tick label style={font=\footnotesize},
    grid=major,
    major grid style={draw=gray!28},
    axis on top,
    axis line style={draw=gray!70},
    tick style={draw=gray!70},
    legend style={
      draw=black!18,
      fill=white,
      fill opacity=0.86,
      text opacity=1,
      font=\DocumentLegendFont,
      at={(0.98,0.98)},
      anchor=north east
    },
    legend cell align={left}
  ]
    \addplot+[thick, no marks, color=black] coordinates {
      (30,19.88293168)
      (40,14.91219876)
      (50,11.92975901)
      (60,9.941465838)
      (70,8.521256433)
      (80,7.456099379)
      (90,6.627643892)
      (100,5.964879503)
      (110,5.42261773)
      (120,4.970732919)
    };
    \addlegendentry{Model \(\mu_m\)}

    \addplot+[
      thick,
      color=plotD,
      mark=*,
      mark size=1.7pt,
      mark options={solid, fill=plotD},
      error bars/.cd,
        y dir=both,
        y explicit
    ] coordinates {
      (30,15.07) +- (0,1.41954218)
      (40,12.46) +- (0,1.385063175)
      (50,10.505) +- (0,1.322866206)
      (60,8.985) +- (0,1.129059343)
      (70,8.1) +- (0,1.166190379)
      (80,7.49) +- (0,1.081619157)
      (90,6.93) +- (0,1.262972684)
      (100,6.34) +- (0,0.997196069)
      (110,6.005) +- (0,0.8746284926)
      (120,5.67) +- (0,0.7688302804)
    };
    \addlegendentry{\acs{OMP} detections}
  \end{axis}
\end{tikzpicture}

%% file: result/per_bs_results_summary.tex
\begin{tabular*}{\linewidth}{@{}l@{\extracolsep{\fill}}rrrrrr@{}}
\toprule
Method & GOSPA & Missed/trial & False/trial & \shortstack{Loc. RMSE\\{[m]}} & \shortstack{Extent RMSE\\{[m]}} & \shortstack{Heading RMSE\\{[$^\circ$]}} \\
\midrule
\multicolumn{7}{@{}l}{\textit{Coordinated \grbp{} references}} \\
\grbpv{CS} & $2.02 \pm 0.27$ & $4.41 \pm 1.11$ & $0.52 \pm 0.34$ & $1.28 \pm 0.15$ & $0.32 \pm 0.05$ & $16.41 \pm 6.31$ \\
\grbpv{CP} & $1.99 \pm 0.26$ & $3.71 \pm 0.73$ & $0.93 \pm 0.63$ & $1.29 \pm 0.16$ & $0.33 \pm 0.05$ & $15.55 \pm 6.02$ \\
\grbpv{CO} & $\mathbf{1.61 \pm 0.12}$ & $\mathbf{2.50 \pm 0.46}$ & $\mathbf{0.42 \pm 0.30}$ & $\mathbf{1.06 \pm 0.12}$ & $\mathbf{0.28 \pm 0.05}$ & $\mathbf{13.30 \pm 4.93}$ \\
\midrule
\multicolumn{7}{@{}l}{\textit{Uncoordinated distributed \grbp{}}} \\
\grbpv{D} & $3.50 \pm 0.57$ & $13.61 \pm 5.10$ & $2.27 \pm 0.85$ & $2.12 \pm 0.34$ & $0.52 \pm 0.05$ & $25.76 \pm 5.48$ \\
\grbpv{DO} & $3.03 \pm 0.51$ & $10.32 \pm 4.04$ & $1.27 \pm 0.82$ & $1.89 \pm 0.32$ & $0.53 \pm 0.06$ & $24.40 \pm 4.72$ \\
\midrule
\multicolumn{7}{@{}l}{\textit{Proposed \grbp{} handover}} \\
\grbpv{H} & $2.80 \pm 0.48$ & $8.07 \pm 3.04$ & $1.85 \pm 0.65$ & $1.74 \pm 0.27$ & $0.39 \pm 0.06$ & $18.12 \pm 5.74$ \\
\grbpv{HM} & $2.79 \pm 0.47$ & $9.13 \pm 3.38$ & $1.64 \pm 0.57$ & $1.60 \pm 0.21$ & $0.38 \pm 0.05$ & $18.28 \pm 5.72$ \\
\grbpv{HMO} & $\mathbf{2.24 \pm 0.29}$ & $\mathbf{4.93 \pm 1.59}$ & $\mathbf{1.16 \pm 0.34}$ & $\mathbf{1.42 \pm 0.17}$ & $\mathbf{0.35 \pm 0.04}$ & $\mathbf{14.84 \pm 4.27}$ \\
\grbpv{HP} & $2.59 \pm 0.38$ & $7.60 \pm 2.47$ & $1.18 \pm 0.30$ & $1.55 \pm 0.19$ & $0.37 \pm 0.05$ & $17.51 \pm 5.41$ \\
\grbpv{HL} & $2.65 \pm 0.41$ & $8.08 \pm 2.75$ & $1.35 \pm 0.45$ & $1.55 \pm 0.19$ & $0.38 \pm 0.05$ & $18.12 \pm 5.54$ \\
\midrule
\multicolumn{7}{@{}l}{\textit{\revGrBP{State-of-the-art BP-based ETT method (distributed)}}} \\
ETT & $2.86 \pm 0.46$ & $12.29 \pm 5.20$ & $0.63 \pm 0.39$ & $1.37 \pm 0.16$ & $0.52 \pm 0.03$ & $9.37 \pm 0.37$ \\
\bottomrule
\end{tabular*}

%% file: Drawings/all_bs_pm_gospa_timeseries.tikz
\begin{tikzpicture}
    \begin{axis}[
      width=0.92\columnwidth, height=0.42\columnwidth,
      scale only axis,
      xmin=0, xmax=100, xtick={0,20,40,60,80,100},
      xlabel={Frame}, ylabel={GOSPA}, ymin=0, ymax=7.2,
      ytick={0,2,4,6},
      label style={font=\scriptsize}, tick label style={font=\scriptsize},
      unbounded coords=jump
    ]
      \let\FirstActiveTrackSeen\relax
      \addplot[thick, solid, no marks, smooth, color=plotC, start from first active monte carlo track]
        table[x=frame,y=gospa_root_mean_power_m_centralized,col sep=comma]{result/all_bs_metrics_timeseries_GOSPA.csv};
      \let\FirstActiveTrackSeen\relax
      \addplot[thick, solid, no marks, smooth, color=plotCL, start from first active monte carlo track]
        table[x=frame,y=gospa_root_mean_power_m_centralized_likelihood_fusion,col sep=comma]{result/all_bs_metrics_timeseries_GOSPA.csv};
      \let\FirstActiveTrackSeen\relax
      \addplot[thick, dashed, no marks, smooth, color=plotCO, start from first active monte carlo track]
        table[x=frame,y=gospa_root_mean_power_m_centralized_true_m,col sep=comma]{result/all_bs_metrics_timeseries_GOSPA.csv};
      \let\FirstActiveTrackSeen\relax
      \addplot[thick, solid, no marks, smooth, color=plotD, start from first active monte carlo track]
        table[x=frame,y=gospa_root_mean_power_m_distributed,col sep=comma]{result/all_bs_metrics_timeseries_GOSPA.csv};
      \let\FirstActiveTrackSeen\relax
      \addplot[thick, dashed, no marks, smooth, color=plotDO, start from first active monte carlo track]
        table[x=frame,y=gospa_root_mean_power_m_distributed_true_m,col sep=comma]{result/all_bs_metrics_timeseries_GOSPA.csv};
      \let\FirstActiveTrackSeen\relax
      \addplot[thick, densely dotted, no marks, smooth, color=plotETT, start from first active monte carlo track]
        table[x=frame,y=gospa_root_mean_power_m_distributed_eot_particle,col sep=comma]{result/all_bs_metrics_timeseries_GOSPA.csv};
      \let\FirstActiveTrackSeen\relax
      \addplot[thick, solid, no marks, smooth, color=plotHM, start from first active monte carlo track]
        table[x=frame,y=gospa_root_mean_power_m_handover_without_m,col sep=comma]{result/all_bs_metrics_timeseries_GOSPA.csv};
      \let\FirstActiveTrackSeen\relax
      \addplot[thick, solid, no marks, smooth, color=plotHWM, start from first active monte carlo track]
        table[x=frame,y=gospa_root_mean_power_m_handover_with_m,col sep=comma]{result/all_bs_metrics_timeseries_GOSPA.csv};
      \let\FirstActiveTrackSeen\relax
      \addplot[thick, solid, no marks, smooth, color=plotHP, start from first active monte carlo track]
        table[x=frame,y=gospa_root_mean_power_m_handover_with_posterior,col sep=comma]{result/all_bs_metrics_timeseries_GOSPA.csv};
      \let\FirstActiveTrackSeen\relax
      \addplot[thick, solid, no marks, smooth, color=plotHL, start from first active monte carlo track]
        table[x=frame,y=gospa_root_mean_power_m_handover_with_likelihood,col sep=comma]{result/all_bs_metrics_timeseries_GOSPA.csv};
      \let\FirstActiveTrackSeen\relax
      \addplot[thick, dashed, no marks, smooth, color=plotHMO, start from first active monte carlo track]
        table[x=frame,y=gospa_root_mean_power_m_handover_with_true_m,col sep=comma]{result/all_bs_metrics_timeseries_GOSPA.csv};
    \end{axis}
    \begin{axis}[
      width=0.92\columnwidth, height=0.42\columnwidth,
      scale only axis,
      xmin=0, xmax=100,
      axis y line*=right, axis x line=none,
      ymin=0, ymax=5.2, ylabel={Avg. track count}, ytick pos=right,
      ytick={0,2,4},
      y label style={font=\scriptsize}, yticklabel style={font=\scriptsize},
      unbounded coords=jump
    ]
      \let\FirstActiveTrackSeen\relax
      \addplot[no marks, black, dashed, start from first active monte carlo track]
        table[x=frame,y=track_count_avg,col sep=comma]{result/all_bs_metrics_timeseries_GOSPA.csv};
    \end{axis}
  \end{tikzpicture}

%% file: Drawings/scalability_comm_vs_N.tikz
\begin{tikzpicture}
\begin{semilogyaxis}[
    width=\linewidth,
    height=0.95\linewidth,
    scale only axis,
    xlabel={Number of \acp{BS} $N$},
    ylabel={Max per-node comm.\ [KB/frame]},
    xmin=0, xmax=65,
    xtick={1,7,19,37,61},
    ymin=0.1, ymax=100,
    grid=both,
    grid style={gray!20},
    tick label style={font=\DocumentLegendFont},
    label style={font=\DocumentLegendFont},
    legend style={at={(0.02,0.98)}, anchor=north west, draw=black!30, fill=white, fill opacity=0.78, text opacity=1, font=\DocumentLegendFont},
    legend cell align=left,
]
\addplot[thick, color=plotC, mark=*, mark size=1.2pt] coordinates {
  (1,1.28125) (7,7.89062) (19,20.7188) (37,41.6094) (61,69.1094)
};
\addlegendentry{\grbpv{CS} (fusion node)}
\addplot[thick, color=plotHM, mark=triangle*, mark size=1.4pt] coordinates {
  (7,0.382812) (19,0.574219) (37,0.765625) (61,0.574219)
};
\addlegendentry{\grbpv{H}}
\addplot[thick, color=plotHWM, mark=diamond*, mark size=1.4pt] coordinates {
  (7,0.601562) (19,0.699219) (37,0.96875) (61,0.902344)
};
\addlegendentry{\grbpv{HM}}
\addplot[thick, color=plotHL, mark=square*, mark size=1.2pt] coordinates {
  (7,0.765625) (19,1.14844) (37,1.53125) (61,1.14844)
};
\addlegendentry{\grbpv{HL}/\grbpv{HP}}
\end{semilogyaxis}
\end{tikzpicture}

%% file: Drawings/scalability_comp_vs_N.tikz
\begin{tikzpicture}
\begin{semilogyaxis}[
    width=\linewidth,
    height=0.95\linewidth,
    scale only axis,
    xlabel={Number of \acp{BS} $N$},
    ylabel={Estimated Max per-node runtime [ms/frame]},
    xmin=0, xmax=65,
    xtick={1,7,19,37,61},
    ymax=700,
    grid=both,
    grid style={gray!20},
    tick label style={font=\DocumentLegendFont},
    label style={font=\DocumentLegendFont},
    legend style={at={(0.02,0.98)}, anchor=north west, draw=black!30, fill=white, fill opacity=0.78, text opacity=1, font=\DocumentLegendFont},
    legend cell align=left,
]
\addplot[thick, color=plotC, mark=*, mark size=1.2pt] coordinates {
  (1,11.2) (7,71.5) (19,189.6) (37,377.3) (61,625.3)
};
\addlegendentry{\grbpv{CS} (fusion node)}
\addplot[thick, color=plotHL, mark=square*, mark size=1.2pt] coordinates {
  (1,11.2) (7,12.4143) (19,13.2789) (37,14.5973) (61,13.5508)
};
\addlegendentry{\grbpv{HL}/\grbpv{HP}}
\addplot[thick, color=plotHWM, mark=diamond*, mark size=1.4pt] coordinates {
  (1,11.2) (7,14.2143) (19,15.9789) (37,18.1973) (61,16.2508)
};
\addlegendentry{\grbpv{HM}}
\addplot[thick, color=plotHM, mark=triangle*, mark size=1.4pt] coordinates {
  (1,11.2) (7,12.2143) (19,12.9789) (37,14.1973) (61,13.2508)
};
\addlegendentry{\grbpv{H}}
\addplot[thick, color=plotD, mark=triangle*, mark size=1.4pt] coordinates {
  (1,11.2) (7,10.2143) (19,9.97895) (37,10.1973) (61,10.2508)
};
\addlegendentry{\grbpv{D}}
\end{semilogyaxis}
\end{tikzpicture}

%% file: result/handover_counting_summary.tex
\setlength{\tabcolsep}{2pt}
\begin{tabular}{@{}lrrrrr@{}}
\toprule
Method & \shortstack{Meas.} & \shortstack{Prior} & \shortstack{Evidence\\ sent} & \shortstack{Fused\\ return} & \shortstack{Total comm.} \\
Unit & 1000 & 1000 & 1000 & 1000 & KB \\
\midrule
\grbpv{CS}/\grbpv{CP} & 29.98 & -- & -- & -- & 468 \\
\grbpv{H} & -- & 0.17 & -- & -- & 32 \\
\grbpv{HM} & -- & 0.17 & 4.2 & 0.33 & 161 \\
\grbpv{HP} & -- & 0.17 & 0.99 & 0.64 & 344 \\
\grbpv{HL} & -- & 0.17 & 0.99 & 0.63 & 343 \\
\bottomrule
\end{tabular}

%% file: Drawings/runtime_comparison_eot_vs_pm.tikz
\begin{tikzpicture}
  \begin{semilogyaxis}[
    width=0.90\columnwidth,
    height=0.30\columnwidth,
    scale only axis,
    xmin=0, xmax=99,
    ymin=10, ymax=10000,
    xtick={0,20,40,60,80},
    ytick={10,100,1000,10000},
    xlabel={Frame},
    ylabel={Mean filter runtime [ms/BS-frame]},
    label style={font=\footnotesize},
    tick label style={font=\footnotesize},
    grid=both,
    major grid style={draw=gray!18},
    minor grid style={draw=gray!10},
    axis on top,
    axis line style={draw=gray!60},
    tick style={draw=gray!60},
    legend style={
      draw=none,
      fill=none,
      font=\DocumentLegendFont,
      legend columns=2,
      at={(0.5,1.03)},
      anchor=south,
      /tikz/every even column/.style={column sep=0.8em}
    },
    legend cell align={left},
    unbounded coords=jump
  ]
    \addplot[thick, solid, no marks, color=plotD] coordinates {
      (0,19.0844645295) (1,29.3788808302) (2,29.8983755696) (3,31.2354269069) (4,32.9361495753)
      (5,34.1285973365) (6,34.4667431022) (7,34.6679565734) (8,34.7807101198) (9,35.1974733132)
      (10,36.8774424625) (11,38.1710688698) (12,40.1228918835) (13,39.9329040959) (14,40.4807346523)
      (15,40.3991267126) (16,41.6641858143) (17,41.996330467) (18,42.4912000841) (19,46.9319686274)
      (20,49.8226908842) (21,56.6433742531) (22,57.5394971133) (23,57.9418713616) (24,57.321339225)
      (25,58.2224983636) (26,64.9440415749) (27,67.1790846771) (28,69.239477987) (29,72.7707436164)
      (30,73.2917255467) (31,73.8167754111) (32,75.2245123135) (33,72.3485107974) (34,69.7208686081)
      (35,69.2099622084) (36,73.3146467402) (37,75.7864273496) (38,80.0769452327) (39,80.5216277586)
      (40,80.2229911893) (41,82.5250590817) (42,78.5889167831) (43,80.0507066997) (44,78.4259934525)
      (45,82.0284985967) (46,87.6979849111) (47,89.1811049214) (48,86.2522368206) (49,85.1282683309)
      (50,82.2634681898) (51,80.9460379936) (52,79.4694011412) (53,79.969666327) (54,75.2377350764)
      (55,75.9008264712) (56,80.6518608591) (57,81.9572434964) (58,82.5307606591) (59,85.5946642556)
      (60,85.2784156881) (61,84.6084927037) (62,81.3432479459) (63,81.0581748736) (64,79.2633246483)
      (65,77.3751808489) (66,76.5476156257) (67,79.3221694241) (68,78.5785007932) (69,77.8589528303)
      (70,73.9362398409) (71,71.4772089442) (72,68.6952845674) (73,67.6073541575) (74,67.5694060975)
      (75,65.9527800832) (76,69.4712253932) (77,68.723521283) (78,67.2660920742) (79,66.1240934444)
      (80,67.5760302704) (81,66.3698776456) (82,64.4544745391) (83,65.9186221896) (84,64.802782818)
      (85,61.260141622) (86,60.8934529151) (87,61.8101143745) (88,63.0484230113) (89,64.2682842672)
      (90,62.578688893) (91,66.679629854) (92,62.0598287164) (93,62.6166887733) (94,62.5052078915)
      (95,63.1046802032) (96,63.2105442949) (97,63.6713029605) (98,65.3355644642) (99,68.1798322174)
    };
    \addlegendentry{\grbpv{D}}

    \addplot[thick, dashed, no marks, color=plotC] coordinates {
      (0,883.42952669) (1,895.203183529) (2,891.024073182) (3,982.939654874) (4,1052.79377107)
      (5,1151.35444413) (6,1198.05958374) (7,1245.74669656) (8,1366.32255767) (9,1448.80308274)
      (10,1667.35238624) (11,1887.26351038) (12,2127.6171845) (13,2327.93717628) (14,2450.28159819)
      (15,2585.42105021) (16,2667.50197295) (17,2811.00926708) (18,2945.02551547) (19,2988.87831714)
      (20,3230.87997569) (21,3199.11545582) (22,3226.49834517) (23,3198.55423781) (24,3207.46749817)
      (25,3571.56021067) (26,3828.20293868) (27,4118.63272725) (28,4178.27103484) (29,4261.52589782)
      (30,4347.11339101) (31,4457.00975094) (32,4634.11697635) (33,4579.31383732) (34,4617.63368579)
      (35,4881.78921313) (36,5304.62680314) (37,5433.67169422) (38,5745.62763017) (39,5559.47564594)
      (40,5667.28917598) (41,5514.37736519) (42,5645.80313768) (43,5873.03680985) (44,5864.94469501)
      (45,6208.93908433) (46,6385.62774693) (47,6240.5507983) (48,6109.37176983) (49,6210.57581051)
      (50,6267.31832526) (51,6045.93906361) (52,6093.2824838) (53,5801.37743049) (54,5580.2128521)
      (55,5754.15508097) (56,6049.92419451) (57,5930.61423691) (58,6071.55979164) (59,5931.77948158)
      (60,5834.73986515) (61,5633.03985682) (62,5464.08274265) (63,5459.49273844) (64,5125.6139437)
      (65,4958.27164384) (66,4733.85083563) (67,4869.52263142) (68,4715.77085158) (69,4568.88500314)
      (70,4594.81182833) (71,4591.88331398) (72,4300.5595139) (73,4495.47075395) (74,4286.97361521)
      (75,4498.95604193) (76,4501.40352308) (77,4566.68382551) (78,4571.09955299) (79,4381.64843255)
      (80,4334.2396793) (81,3970.73250712) (82,4027.97618571) (83,4013.02439134) (84,3694.95146186)
      (85,3607.31483105) (86,3569.29183591) (87,3792.72694365) (88,3983.2564579) (89,4030.81528957)
      (90,4120.04098415) (91,4013.04579287) (92,3839.41670309) (93,3908.21972226) (94,3909.26645273)
      (95,3950.4284987) (96,4018.13196282) (97,4055.861776) (98,4112.86219627) (99,4161.08719039)
    };
    \addlegendentry{\revGrBP{ETT (BP, distributed)}}

  \end{semilogyaxis}
\end{tikzpicture}

%% file: Drawings/bp_message_runtime_eot_vs_pm.tikz
\begin{tikzpicture}
  \begin{axis}[
    width=0.90\columnwidth,
    height=0.30\columnwidth,
    scale only axis,
    xmin=0, xmax=99,
    ymin=0, ymax=10,
    xtick={0,20,40,60,80},
    ytick={0,2,4,6,8,10},
    xlabel={Frame},
    ylabel={Mean BP-message time [ms/BS-frame]},
    label style={font=\footnotesize},
    tick label style={font=\footnotesize},
    grid=major,
    major grid style={draw=gray!18},
    axis on top,
    axis line style={draw=gray!60},
    tick style={draw=gray!60},
    legend style={
      draw=none,
      fill=none,
      font=\DocumentLegendFont,
      legend columns=2,
      at={(0.5,1.03)},
      anchor=south,
      /tikz/every even column/.style={column sep=0.8em}
    },
    legend cell align={left},
    unbounded coords=jump
  ]
    \addplot[thick, solid, no marks, color=plotD] coordinates {
      (0,0.0782027968151) (1,0.928429517444) (2,0.969835416654) (3,0.994286139394) (4,1.02238545584)
      (5,1.03620987703) (6,1.01521199976) (7,1.03888709988) (8,1.03398941806) (9,1.03666415291)
      (10,1.05418422705) (11,1.0791372872) (12,1.11072932361) (13,1.08336416905) (14,1.08109307384)
      (15,1.09531086148) (16,1.09730323024) (17,1.11140024314) (18,1.11975412334) (19,1.24931553524)
      (20,1.3145329589) (21,1.3624607625) (22,1.42637728857) (23,1.38952941347) (24,1.35734149378)
      (25,1.38127391381) (26,1.48247581277) (27,1.53718521019) (28,1.54756416981) (29,1.55666248834)
      (30,1.55431394444) (31,1.55406349902) (32,1.56224925459) (33,1.52654576256) (34,1.49345280368)
      (35,1.49870122553) (36,1.53484251481) (37,1.61542150648) (38,1.66094140063) (39,1.67065405183)
      (40,1.68045691093) (41,1.65616947094) (42,1.5678104925) (43,1.68456553759) (44,1.77064649741)
      (45,1.76806884819) (46,1.79669289136) (47,1.76050841438) (48,1.73744551932) (49,1.74514302015)
      (50,1.75397599102) (51,1.62321245306) (52,1.62891456348) (53,1.62252145177) (54,1.56379806115)
      (55,1.55821639599) (56,1.60050681379) (57,1.66593765982) (58,1.66774371606) (59,1.70532971549)
      (60,1.69353650847) (61,1.65308151718) (62,1.63428721026) (63,1.63204070378) (64,1.60018562072)
      (65,1.56332395289) (66,1.59710499139) (67,1.64324934614) (68,1.600175044) (69,1.61402340974)
      (70,1.57782158763) (71,1.52229361709) (72,1.50954057878) (73,1.48906069376) (74,1.46193809225)
      (75,1.47898675145) (76,1.47324737697) (77,1.50738152138) (78,1.49995792633) (79,1.46585119005)
      (80,1.4838355279) (81,1.45383752322) (82,1.4604194752) (83,1.47026089535) (84,1.41502881035)
      (85,1.40168500977) (86,1.4092979724) (87,1.4191648038) (88,1.49531046989) (89,1.53527067644)
      (90,1.51476887331) (91,1.5386035646) (92,1.47005518979) (93,1.57184757968) (94,1.4823719201)
      (95,1.50659185146) (96,1.52767972595) (97,1.53107510195) (98,1.51101508807) (99,1.52069259872)
    };
    \addlegendentry{\grbpv{D}}

    \addplot[thick, dashed, no marks, color=plotC] coordinates {
      (0,3.48058195062) (1,3.53881276385) (2,3.52716341315) (3,3.72712344253) (4,3.85600208345)
      (5,4.06036588023) (6,4.10499991703) (7,4.21721752673) (8,4.45876320939) (9,4.5701900164)
      (10,4.94854451775) (11,5.25354580951) (12,5.59571489525) (13,5.85713294714) (14,5.98174469448)
      (15,6.20050568755) (16,6.32082817359) (17,6.47772939382) (18,6.72903017271) (19,6.92791703522)
      (20,7.13436532565) (21,7.04161164261) (22,7.0559402741) (23,7.00010255695) (24,6.99543717239)
      (25,7.34957666092) (26,7.60095486947) (27,7.88448293172) (28,7.87340521859) (29,7.88110349894)
      (30,8.00696860246) (31,8.04690518735) (32,8.15539042952) (33,7.97356796926) (34,7.96747297349)
      (35,8.21078489548) (36,8.46160427884) (37,8.70033481766) (38,8.9839213713) (39,8.81100690848)
      (40,9.10057943701) (41,9.08534025177) (42,9.21732954924) (43,9.43116294335) (44,9.40708520614)
      (45,9.72251096095) (46,9.8330842471) (47,9.70104344207) (48,9.60866928603) (49,9.62432223544)
      (50,9.71593335738) (51,9.43229090721) (52,9.47875076498) (53,9.17427201303) (54,9.05136164987)
      (55,9.2071578295) (56,9.46542287318) (57,9.37569172158) (58,9.5256125725) (59,9.37899734121)
      (60,9.37966744499) (61,9.17873195233) (62,9.05054808944) (63,9.06184637173) (64,8.84304551254)
      (65,8.68389320649) (66,8.54407805752) (67,8.67522197348) (68,8.47795835655) (69,8.26818985605)
      (70,8.1723367332) (71,8.16275444325) (72,7.76808187412) (73,7.9077851396) (74,7.66332451959)
      (75,7.91731533653) (76,7.82809443689) (77,7.90475022306) (78,7.98800891364) (79,7.77938186404)
      (80,7.8130295892) (81,7.43863189753) (82,7.58748374327) (83,7.64565904202) (84,7.35463355479)
      (85,7.36433613799) (86,7.34422799581) (87,7.56174017435) (88,7.69705314884) (89,7.71260224335)
      (90,7.83267044507) (91,7.70843631782) (92,7.50142383597) (93,7.62522075362) (94,7.65540947399)
      (95,7.72719817695) (96,7.79791141154) (97,7.86729339389) (98,7.83651301642) (99,7.92380015747)
    };
    \addlegendentry{\revGrBP{ETT (BP, distributed)}}

  \end{axis}
\end{tikzpicture}

%% file: sections/sec_conclusion.tex
\vspace{-1mm}
\section{Conclusion}
\label{sec:conclusion}

We proposed \revGrBP{GrBP} and event-triggered owner--shadow handover for \ac{DISAC} with partially overlapping finite \acp{FOV}. Range- and visibility-aware detection and local owner--UID labels avoid continuous centralized fusion and global label synchronization. Prior-only \grbpv{H} reduces boundary track loss and improves on uncoordinated \grbpv{D} with bounded, event-driven communication. \grbpv{HM}, \grbpv{HL}, and \grbpv{HP} improve selected metrics at higher communication cost.

Future work will refine the measurement model and compare random-matrix with alternative extent representations. It will also study tractable multi-partition or sampling-based joint clustering and tracking. Further directions include validating the model with realistic ray tracing and channel-parameter estimation, quantifying accuracy--complexity trade-offs, and deriving sufficient conditions for bounded local workload. Finally, we will develop compressed, asynchronous, failure-resistant handover for realistic sensing, propagation, and mobility.


%% file: sections/sec_appendix_proofs.tex
\vspace{-1mm}
\section{Proof of Proposition~\ref{prop:centralized_nonscalable}}
\label{sec:appendix_scalability}
\setcounter{equation}{0}
\begin{proof}
Consider the centralized coordinated architecture, in which every \ac{BS} forwards its local
measurements to the fusion node \(0\in\Set{V}(\mathcal M)\) that performs the joint
multi-target update at each frame \(k\).
Node~\(0\) receives \(\sum_{s\in\Set{B}}|\V{Z}_{s,k}|\) detections. With \(b_z\)
bytes per detection and \(\mathbb{E}[|\V{Z}_{s,k}|]\ge\mu_{\min}>0\) (each \ac{BS}
contributes at least its clutter), its expected communication load obeys
\(\mathbb{E}[C^{\mathrm{comm}}_{0,k}]\ge b_z\sum_{s}\mathbb E[|\V{Z}_{s,k}|]\ge b_z\mu_{\min}N\),
and any centralized update must read each received detection, so
\(\mathbb{E}[C^{\mathrm{comp}}_{0,k}]=\bl{\Omega(N)}\). Both diverge as \(N\to\infty\),
so no \(N\)-independent constants bound node~\(0\). Since
Definition~\ref{def:scalable} requires bounds at every node and
\(0\in\Set{V}(\mathcal M)\), the centralized coordinated architectures violate
it.
\end{proof}

\vspace{-1mm}
\section{Proof of Proposition~\ref{prop:handover_scalable}}
\label{sec:app_handover_scalability_proofs}

\begin{proof}
Here \(\Set{V}(\mathcal M)=\Set{B}\), with no fusion node. At \ac{BS}~\(s\) the
single-\ac{BS} update costs \(O(n^{t}_{s,k}n^{g}_{s,k})\) \eqref{eq:local_comp_cost}. Evaluating
 the handover score over the \(\le d_{\max}\) neighbors of each track,
maintaining the owner/shadow records, and (for the optional
variants) fusing the \(\le d_{\max}\) packets per shared track add
\(O(n^{t}_{s,k}d_{\max})\), giving
\eqref{eq:handover_comp_bound}. For communication, only owners
send priors and \(\mathcal{R}_o(\ell)\) suppresses repeats within a visibility
episode, so each (track, neighbor) pair sends at most one prior
per threshold upcrossing of
\eqref{eq:prior_handover_trigger}. \ac{BS}~\(s\) participates in at most
\(d_{\max}n^{t}_{s,k}\) such pairs, each carrying \(\le b_{\mathrm{evt}}\) bytes. Accounting
for both directions gives the first term of \eqref{eq:handover_comm_bound}. The
optional variants add at most \(d_{\max}\) active owner--shadow edges per shared
track, each with an uplink and a downlink packet of size \(b_Z+b_Q\),
\(b_\gamma+b_Q\), or \(2b_Q\), giving the second term.
The conditional bounds depend on \(N\) only through the local counts and degree.
Taking expectations, Assumptions~\ref{as:degree}--\ref{as:representation} bound
\(d_{\max}\) and the second moments of \(n^{t}_{s,k},n^{g}_{s,k}\),\footnote{Consistent
with the unbounded-support Poisson models of \secref{sec:detection_rate}: \(r_{\min}\)
caps \(\mu_{m,s}\) via \eqref{eq:rho_from_pd_cell_density} and the bounded region and
clutter intensity cap the clutter mean, so \(m_{s,k}\) (hence \(n^{g}_{s,k}\le m_{s,k}\))
has finite, \(N\)-independent moments despite unbounded support.} and Cauchy--Schwarz
bounds \(\mathbb E[n^{t}_{s,k}n^{g}_{s,k}]\) and \(\mathbb E[n^{t}_{s,k}]\) by
\(N\)-independent constants. Hence all four variants are scalable.
\end{proof}

%% file: sections/sec_appendixI.tex
\section*{S-A. Derivation of the Approximate Joint PDF}
\phantomsection
\makeatletter
\def\@currentlabel{S-A}
\makeatother
\label{sec:appendix_I}

\revGrBP{We derive the single-step joint \ac{PDF} and its factorization for grouped-measurement belief propagation (GrBP), the local multi-target tracking method of \secref{main-sec:single_bs_method}, treating the external partition as given side information.}\acused{GrBP} Approximations are marked by \(\approx\): the clustering-induced group count is modeled as Poisson, so the resulting factorization is an approximate conditional model rather than a generative law over the raw measurements. The derivation follows the point-target message-passing filters \cite{meyer2018messagepassing,meyer2018messagepassingSupplement}, with two differences specific to our extended-object model: measurements are processed in groups rather than one per target, and the number of clutter groups produced by the clustering step is approximated by a Poisson distribution. The key idea is to express the unknown numbers of detected legacy targets, newborn targets, and clutter groups as deterministic functions of the association and existence variables, after which the joint \ac{PDF} factorizes into per-target terms. We omit the time index \(k\) throughout.

\begin{figure*}[!t]
  \centering
  \resizebox{\textwidth}{!}{\input{Drawings/graph_drawing}}
  \caption{Factor graph of \revGrBP{GrBP}, with the time index \(k\) shown explicitly. Circles denote state and association variables. Boxes denote the legacy priors, the group likelihood factors in \eqref{eq:legacy_l_function} and \eqref{eq:new_l_function}, and the association-consistency factor \(\boldsymbol{\psi}\).}
  \label{fig:supp_factor_graph}
\end{figure*}
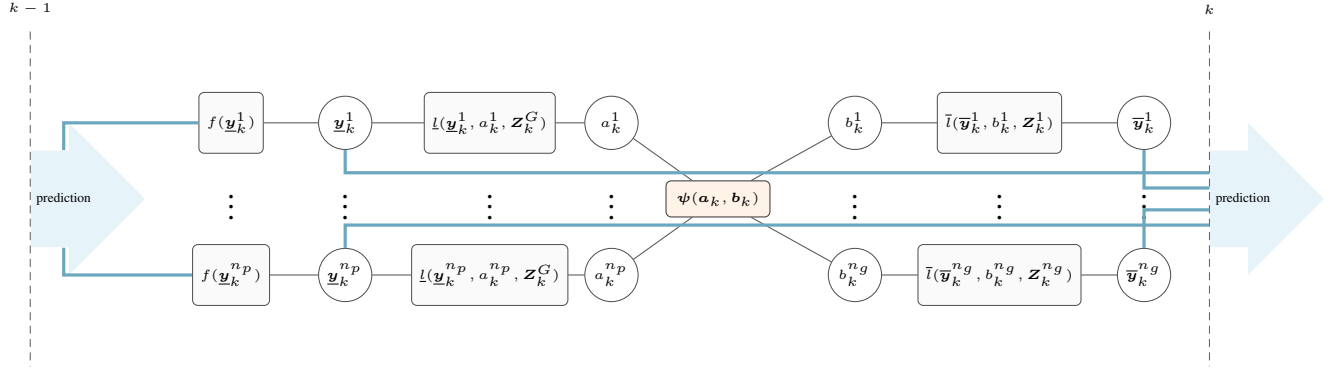

\newcommand{\Pm}[2]{{}_{#1}P_{#2}}  
\newcommand{\mucgroup}{\mu_{\mathrm{g}}}
\vspace{-1mm}
\subsection{Notation and Setup}

We reuse the notation of \secref{main-sec:system_model} and \secref{main-sec:single_bs_method}. The \(m\) measurements are grouped into \(\mathrm{n}^\mathrm{g}\) groups \(\V{Z}^{G}=\{\V{Z}^1,\dots,\V{Z}^{\mathrm{n}^\mathrm{g}}\}\), where group \(\V{Z}^g\) contains \(|\V{Z}^g|\) measurements and both \(m\) and \(\mathrm{n}^\mathrm{g}\) are determined by \(\V{Z}^{G}\). We use ``group'' rather than ``cluster'' to reserve the superscript \(c\) for clutter. The \(\mathrm{n}^\mathrm{t}\) legacy targets have augmented states \(\underline{\V{y}}^i=[\underline{\V{x}}^{i\T},\underline{\M{E}}^{i\T},\underline{r}^i]^\T\) stacked in \(\underline{\V{Y}}\). We introduce one newborn component per group, \(\overline{\V{y}}^j=[\overline{\V{x}}^{j\T},\overline{\M{E}}^{j\T},\overline{r}^j]^\T\) stacked in \(\overline{\V{Y}}\), write \(\overline{\Set{X}}=((\overline{\V{x}}^1,\overline{\M{E}}^1),\dots,(\overline{\V{x}}^{\mathrm{n}^\mathrm{g}},\overline{\M{E}}^{\mathrm{n}^\mathrm{g}}))\) for the newborn kinematic--extent pairs, and collect the legacy and newborn existence indicators in \(\underline{\V{r}}\) and \(\overline{\V{r}}\).
Assuming, as in \cite[Eq.~(36)]{meyer2018messagepassing} and \cite[Eqs.~(1)--(4)]{meyer2021scalable}, that the time-\((k{-}1)\) posterior and the transition density factorize into single-target terms, the legacy prior is
\begin{equation}
f(\underline{\V{Y}}) = \prod_{i=1}^{\mathrm{n}^\mathrm{t}} f(\underline{\V{y}}^i),
\label{eq:pt_legacy_prior}
\end{equation}
with \(f(\underline{\V{y}}^i)\) the predicted marginal of legacy target \(i\).

Following \cite[Eqs.~(10),(20)]{meyer2018messagepassing}, the target- and measurement-oriented association vectors are \(\V{a}=(a^1,\dots,a^{\mathrm{n}^\mathrm{t}})\) and \(\V{b}=(b^1,\dots,b^{\mathrm{n}^\mathrm{g}})\): \(a^i=g\) if legacy target \(i\) generates group \(g\) and \(a^i=0\) otherwise, while \(b^j=\ell\) if group \(j\) is generated by legacy target \(\ell\) and \(b^j=0\) otherwise. The consistency factor \(\boldsymbol{\psi}(\V{a},\V{b})=\prod_{i,j}\Psi_{i,j}(a^i,b^j)\), with \(\Psi_{i,j}=0\) iff exactly one of \(a^i=j\) and \(b^j=i\) holds and \(\Psi_{i,j}=1\) otherwise, restricts the model to exclusive association hypotheses. Its marginals reproduce the single-sided exclusivity factors, \(\sum_{\V{b}}\boldsymbol{\psi}(\V{a},\V{b})=\boldsymbol{\psi}(\V{a})\) and \(\sum_{\V{a}}\boldsymbol{\psi}(\V{a},\V{b})=\boldsymbol{\psi}(\V{b})\).

Each group originates from a single source. We collect the indices of detected legacy targets, existing newborns, and clutter groups in \(\mathbb{D}(\V{a})=\{i:a^i\neq0\}\), \(\mathbb{N}(\overline{\V{r}})=\{j:\overline{r}^j=1\}\), and \(\mathbb{C}(\V{a},\overline{\V{r}})=\{1,\dots,\mathrm{n}^\mathrm{g}\}\setminus(\{a^i:i\in\mathbb{D}(\V{a})\}\cup\mathbb{N}(\overline{\V{r}}))\), with cardinalities \(\mathrm{n}^\mathrm{d}=|\mathbb{D}(\V{a})|\), \(\mathrm{n}^\mathrm{n}=|\mathbb{N}(\overline{\V{r}})|\), and \(\mathrm{n}^\mathrm{c}=|\mathbb{C}(\V{a},\overline{\V{r}})|=\mathrm{n}^\mathrm{g}-\mathrm{n}^\mathrm{d}-\mathrm{n}^\mathrm{n}\). These three counts are deterministic functions of \((\V{a},\overline{\V{r}})\).
\vspace{-1mm}
\subsection{Joint PDF and its Factorization}

By the chain rule, and inserting the exclusivity factor \(\boldsymbol{\psi}(\V{a},\V{b})\), the joint \ac{PDF} decomposes into the legacy prior, an association model, and a likelihood conditioned on the association,\footnote{\color{black}In this appendix, \(f(\cdot,\mathbf Z^G)\) denotes the density conditional on the fixed partition \(\mathcal P=\mathcal C(\mathbf Z)\) supplied by the external clustering algorithm. Dependence on \(\mathcal P\) is suppressed for readability. Equation~\eqref{eq:joint_pdf_structured} is therefore an exact chain-rule decomposition of this conditional density. The subsequent Poisson group-count and conditional group-independence models approximate its factors. We do not claim to model the probability that the clustering algorithm produces \(\mathcal P\). Moreover, the conditional model assumes that each target generates at most one group and each group contains measurements from at most one target. Clustering splits and cross-target merges are not represented.}
\begin{equation}
\begin{aligned}
f(\underline{\V{Y}},\overline{\V{Y}},\V{a},\V{b},\V{Z}^{G})
={}& f(\underline{\V{Y}})\,
\boldsymbol{\psi}(\V{a},\V{b})\,
f(\V{a},\V{b},\overline{\V{Y}}\mid\underline{\V{Y}})\\
&\times f(\V{Z}^{G}\mid\underline{\V{Y}},\overline{\V{Y}},\V{a},\V{b}).
\end{aligned}
\label{eq:joint_pdf_structured}
\end{equation}
The legacy prior is \eqref{eq:pt_legacy_prior}. We factorize the association model and the likelihood in turn, omitting \(\boldsymbol{\psi}(\V{a},\V{b})\) until the final result.

\subsubsection{Association model} We split the association model into an existence-and-association part and the newborn densities, \(f(\V{a},\V{b},\overline{\V{Y}}\mid\underline{\V{Y}})=f(\overline{\V{r}},\V{a},\V{b}\mid\underline{\V{Y}})\,f(\overline{\Set{X}}\mid\overline{\V{r}})\). The numbers of clutter groups and detected newborns are modeled as Poisson \cite[Eqs.~(3),(18)]{meyer2018messagepassingSupplement}, \(f(\mathrm{n}^\mathrm{c})\approx\mucgroup^{\,\mathrm{n}^\mathrm{c}}\mathrm e^{-\mucgroup}/\mathrm{n}^\mathrm{c}!\) and \(f(\mathrm{n}^\mathrm{n})=\mu_\mathrm{n}^{\,\mathrm{n}^\mathrm{n}}\mathrm e^{-\mu_\mathrm{n}}/\mathrm{n}^\mathrm{n}!\), where \(\mu_\mathrm{n}\) is the mean number of detected newborns. The clustering step turns the clutter \ac{PPP} of mean \(\mu_c\) into groups. Since one group may contain several clutter detections, \(\mucgroup\neq\mu_c\) in general, and we approximate the induced group count as Poisson with mean \(\mucgroup\). Conditioned on the counts, the newborn existence pattern is uniform, \(f(\overline{\V{r}}\mid\cdot)=\binom{\mathrm{n}^\mathrm{c}+\mathrm{n}^\mathrm{n}}{\mathrm{n}^\mathrm{n}}^{-1}\), and the associations are uniform over the \(\Pm{\mathrm{n}^\mathrm{g}}{\mathrm{n}^\mathrm{d}}\) feasible assignments, \(f(\V{a},\V{b}\mid\cdot)=1/\Pm{\mathrm{n}^\mathrm{g}}{\mathrm{n}^\mathrm{d}}=(\mathrm{n}^\mathrm{c}+\mathrm{n}^\mathrm{n})!/\mathrm{n}^\mathrm{g}!\) \cite[Eqs.~(22)--(24)]{meyer2018messagepassingSupplement}. A legacy target contributes \(\underline{r}^\mathrm{p}\) when detected and \(1-\underline{r}^\mathrm{p}(1-\mathrm e^{-\mu_{\mathrm{m}}})\) when not, where \(\mu_{\mathrm{m}}=\mu_{\mathrm{m}}(\V{x},\M{E})\) is the expected number of target measurements \eqref{main-eq:system_measurement_count_mean}. The newborn densities are \(\prod_{j\in\mathbb{N}(\overline{\V{r}})}f_\mathrm{n}(\overline{\V{x}}^j,\overline{\M{E}}^j)\prod_{j\notin\mathbb{N}(\overline{\V{r}})}f_d(\overline{\V{x}}^j,\overline{\M{E}}^j)\), with \(f_\mathrm{n}\) the newborn density and \(f_d\) a dummy density that integrates to one. Multiplying these factors, the \(\mathrm{n}^\mathrm{c}!\), \(\mathrm{n}^\mathrm{n}!\), and \((\mathrm{n}^\mathrm{c}+\mathrm{n}^\mathrm{n})!\) terms cancel, and the association model \revGrammar{is} 
\begin{equation}
\begin{aligned}
&f(\V{a},\V{b},\overline{\V{Y}}\mid\underline{\V{Y}})
\approx \frac{\mucgroup^{\,\mathrm{n}^\mathrm{c}}\mathrm e^{-\mucgroup}\,\mu_\mathrm{n}^{\,\mathrm{n}^\mathrm{n}}\mathrm e^{-\mu_\mathrm{n}}}{\mathrm{n}^\mathrm{g}!}
\\
&\quad\times
\prod_{p\in\mathbb{D}(\V{a})}\underline{r}^{\mathrm{p}}
\prod_{p\notin\mathbb{D}(\V{a})}\bigl[1-\underline{r}^{\mathrm{p}}(1-\mathrm e^{-\mu_{\mathrm{m}}})\bigr]
\\
&\quad\times
\prod_{j\in\mathbb{N}(\overline{\V{r}})}f_\mathrm{n}(\overline{\V{x}}^{j},\overline{\M{E}}^{j})
\prod_{j\notin\mathbb{N}(\overline{\V{r}})}f_d(\overline{\V{x}}^{j},\overline{\M{E}}^{j}).
\end{aligned}
\label{eq:I13}
\end{equation}

\subsubsection{Likelihood} Conditioned on the association, the groups are generated independently,
\begin{equation}
\begin{aligned}
&f(\V{Z}^{G}\mid\underline{\V{Y}},\overline{\V{Y}},\V{a},\V{b})\\
&\quad=\prod_{p\in\mathbb{D}(\V{a})} f(\V{Z}^{a^{p}}\mid\underline{\V{x}}^{p},\underline{\M{E}}^{p})\\
&\qquad\times\prod_{j\in\mathbb{N}(\overline{\V{r}})} f(\V{Z}^{j}\mid\overline{\V{x}}^{j},\overline{\M{E}}^{j})
\prod_{c\in\mathbb{C}(\V{a},\overline{\V{r}})} f_c(\V{Z}^{c}).
\end{aligned}
\label{eq:pt_group_like_count_split}
\end{equation}
A detected legacy target follows the standard extended-target \ac{PPP} model \cite[Eq.~(24)]{Yuxuanxia_standardmodel}, \(f(\V{Z}\mid\V{x},\M{E})=\mathrm e^{-\mu_{\mathrm{m}}}\prod_{\V{z}\in\V{Z}}\mu_{\mathrm{m}} f(\V{z}\mid\V{x},\M{E})\). A newborn is detected by assumption, so its likelihood is the zero-truncated version \(f(\V{Z}\mid\V{x},\M{E})=[\mathrm e^{-\mu_{\mathrm{m}}}/(1-\mathrm e^{-\mu_{\mathrm{m}}})]\prod_{\V{z}\in\V{Z}}\mu_{\mathrm{m}} f(\V{z}\mid\V{x},\M{E})\). A clutter group is i.i.d., \(f_c(\V{Z}^{c})=\prod_{\V{z}\in\V{Z}^{c}}f_c(\V{z})\), and since \(\mathbb{D}\), \(\mathbb{N}\), and \(\mathbb{C}\) partition the groups, \(\prod_{c=1}^{\mathrm{n}^\mathrm{g}}f_c(\V{Z}^{c})=\prod_{i=1}^{m}f_c(\V{z}^{i})\).

\subsubsection{Factorized joint PDF} Combining the prior \eqref{eq:pt_legacy_prior}, the association model \eqref{eq:I13}, and the likelihood \eqref{eq:pt_group_like_count_split}, dividing each detected group by \(f_c\), and redistributing \(\mucgroup^{\,\mathrm{n}^\mathrm{c}}=\mucgroup^{\,\mathrm{n}^\mathrm{g}}\prod_{p\in\mathbb{D}(\V{a})}\mucgroup^{-1}\prod_{j\in\mathbb{N}(\overline{\V{r}})}\mucgroup^{-1}\), all terms independent of the unknowns collect into the constant
\begin{equation}
C(\mucgroup,\mu_\mathrm{n},\V{Z}^{G})
\triangleq
\frac{\mucgroup^{\,\mathrm{n}^\mathrm{g}}\,\mathrm e^{-\mucgroup}\,\mathrm e^{-\mu_\mathrm{n}}}{\mathrm{n}^\mathrm{g}!}
\prod_{i=1}^{m} f_c(\V{z}^{i}),
\label{eq:pt_constant_C}
\end{equation}
where \(\mathrm{n}^\mathrm{g}\) and \(m\) are determined by \(\V{Z}^{G}\). To lighten the per-target factors, we suppress the component state, writing \(\mu_{\mathrm{m}}\) for \(\mu_{\mathrm{m}}(\V{x},\M{E})\), \(f(\V{z})\) for \(f(\V{z}\mid\V{x},\M{E})\), and \(f_\mathrm{n},f_d\) for \(f_\mathrm{n}(\overline{\V{x}}^{j},\overline{\M{E}}^{j}),f_d(\overline{\V{x}}^{j},\overline{\M{E}}^{j})\). The legacy likelihood factor is
\begingroup
\footnotesize
\setlength{\abovedisplayskip}{1pt plus 1pt minus 1pt}
\setlength{\belowdisplayskip}{1pt plus 1pt minus 1pt}
\setlength{\abovedisplayshortskip}{1pt plus 1pt minus 1pt}
\setlength{\belowdisplayshortskip}{1pt plus 1pt minus 1pt}
\setlength{\jot}{0pt}
\renewcommand{\arraystretch}{0.82}
\begin{equation}
\underline{l}(\underline{\V{y}}^{p},a^{p}\mathpunct{,}\V{Z}^{G})
=
\begin{cases}
\begin{aligned}[t]
&\dfrac{\mathrm e^{-\mu_{\mathrm{m}}}}{\mucgroup}\\
&\times\prod_{\V{z}\in\V{Z}^{a^{p}}}\dfrac{\mu_{\mathrm{m}} f(\V{z})}{f_c(\V{z})},
\end{aligned}
& a^{p}\neq 0,\ \underline{r}^{\mathrm{p}}=1,\\[0.3ex]
\mathrm e^{-\mu_{\mathrm{m}}}, & a^{p}=0,\ \underline{r}^{\mathrm{p}}=1,\\[0.2ex]
1, & a^{p}=0,\ \underline{r}^{\mathrm{p}}=0,\\[0.2ex]
0, & a^{p}\neq 0,\ \underline{r}^{\mathrm{p}}=0,
\end{cases}
\label{eq:legacy_l_function}
\end{equation}
\endgroup
and the newborn likelihood factor is
\begingroup
\footnotesize
\setlength{\abovedisplayskip}{1pt plus 1pt minus 1pt}
\setlength{\belowdisplayskip}{1pt plus 1pt minus 1pt}
\setlength{\abovedisplayshortskip}{1pt plus 1pt minus 1pt}
\setlength{\belowdisplayshortskip}{1pt plus 1pt minus 1pt}
\setlength{\jot}{0pt}
\renewcommand{\arraystretch}{0.82}
\begin{equation}
\overline{l}(\overline{\V{y}}^{j},b^{j}\mathpunct{,}\V{Z}^{j})
=
\begin{cases}
\begin{aligned}[t]
&\dfrac{\mu_\mathrm{n} f_\mathrm{n}\,\mathrm e^{-\mu_{\mathrm{m}}}}{\mucgroup(1-\mathrm e^{-\mu_{\mathrm{m}}})}\\
&\times\prod_{\V{z}\in\V{Z}^{j}}\dfrac{\mu_{\mathrm{m}} f(\V{z})}{f_c(\V{z})},
\end{aligned}
& b^{j}=0,\ \overline{r}^{j}=1,\\[0.3ex]
0, & b^{j}\neq 0,\ \overline{r}^{j}=1,\\[0.2ex]
f_d, & \overline{r}^{j}=0.
\end{cases}
\label{eq:new_l_function}
\end{equation}
\endgroup
With these definitions, the joint \ac{PDF} factorizes into the per-target product used in \secref{main-sec:single_bs_method}.
The corresponding factor graph is shown in \figref{fig:supp_factor_graph}.

\subsection{Sum-Product Messages and Beliefs}
\label{sec:appendix_beliefs}
For the detailed message recursions and posterior-belief calculations, we
refer the reader to \cite[Secs.~VI and IX-A]{meyer2018messagepassing}, with
the corrections in \cite{meyer2018messagepassingSupplement}. We apply the
same sum-product schedule to the bipartite association subgraph, with
measurement groups in place of individual measurements and augmented
kinematic, extent, and existence states in place of point-target states.
The legacy prior and the group likelihood factors are given by
\eqref{eq:pt_legacy_prior}, \eqref{eq:legacy_l_function}, and
\eqref{eq:new_l_function}, respectively. There is one newborn component per
group, whose birth density is already included in \(\overline l\).

%% file: Drawings/graph_drawing.tex
\definecolor{predarrow}{HTML}{DDEFF7} 
\definecolor{predline}{HTML}{6BA7BE}  

\newcommand{\predarrowheight}{15mm} 
\newcommand{\predarrowwidth}{20mm}  

\begin{tikzpicture}[
    >=Stealth,
    every node/.style={font=\tiny},
    obs/.style   ={circle,draw=black!70,fill=white,inner sep=1pt,minimum size=7mm,line width=0.35pt},
    latent/.style={circle,draw=black!70,fill=white,inner sep=1pt,minimum size=7mm,line width=0.35pt},
    box/.style   ={draw=black!70,fill=black!2,rounded corners=2pt,inner sep=3pt,minimum height=8mm,line width=0.35pt},
    pot/.style   ={draw=black!75,fill=plotD!10,rounded corners=2pt,inner sep=4pt,line width=0.4pt},
    predarrow/.style={
        single arrow,
        draw=none,
        fill=predarrow,
        fill opacity=0.72,
        minimum height=\predarrowheight,
        minimum width=\predarrowwidth,
        single arrow head extend=0.7\predarrowheight
    }
]

\node[box]   (f1)  at (0,   1) {$f(\underline{\V{y}}_k^{1})$};
\node[box]   (fp)  at (0,  -1.0) {$f(\underline{\V{y}}_k^{n_p})$};

\node[obs]   (y1)  at (1.5,  1) {$\underline{\V{y}}_k^{1}$};
\node[obs]   (yp)  at (1.5, -1.0) {$\underline{\V{y}}_k^{n_p}$};

\node[box]   (q1)  at (3.4,  1) {$\underline{l}(\underline{\V{y}}_k^{1},a_k^{1}\mathpunct{,}\V{Z}_k^{G})$};
\node[box]   (qp)  at (3.4, -1.0) {$\underline{l}(\underline{\V{y}}_k^{n_p},a_k^{n_p}\mathpunct{,}\V{Z}_k^{G})$};

\node[latent] (a1) at (5.0,  1) {$a_k^{1}$};
\node[latent] (ap) at (5.0, -1.0) {$a_k^{n_p}$};

\node[pot]    (psi) at (6.4,0) {$\boldsymbol{\psi}(\V{a}_k,\V{b}_k)$};

\node[latent] (b1)   at (8.2,  1) {$b_k^{1}$};
\node[latent] (bm)   at (8.2, -1.0) {$b_k^{n_g}$};

\node[box]    (v1)   at (10.1,  1) {$\overline{l}(\overline{\V{y}}_k^{1},b_k^{1}\mathpunct{,}\V{Z}_k^{1})$};
\node[box]    (vm)   at (10.1, -1.0) {$\overline{l}(\overline{\V{y}}_k^{n_g},b_k^{n_g}\mathpunct{,}\V{Z}_k^{n_g})$};

\node[obs]    (ybar1) at (12.0,  1) {$\overline{\V{y}}_k^{1}$};
\node[obs]    (ybarm) at (12.0, -1.0) {$\overline{\V{y}}_k^{n_g}$};

\draw[draw=black!65,line cap=round] (f1) -- (y1);
\draw[draw=black!65,line cap=round] (fp) -- (yp);

\draw[draw=black!65,line cap=round] (y1) -- (q1);
\draw[draw=black!65,line cap=round] (yp) -- (qp);

\draw[draw=black!65,line cap=round] (q1) -- (a1);
\draw[draw=black!65,line cap=round] (qp) -- (ap);

\draw[draw=black!65,line cap=round] (a1) -- (psi);
\draw[draw=black!65,line cap=round] (ap) -- (psi);

\draw[draw=black!65,line cap=round] (psi) -- (b1);
\draw[draw=black!65,line cap=round] (psi) -- (bm);

\draw[draw=black!65,line cap=round] (b1) -- (v1);
\draw[draw=black!65,line cap=round] (bm) -- (vm);

\draw[draw=black!65,line cap=round] (v1) -- (ybar1);
\draw[draw=black!65,line cap=round] (vm) -- (ybarm);

\coordinate (midF)    at ($(f1)!0.5!(fp)$);
\coordinate (midY)    at ($(y1)!0.5!(yp)$);
\coordinate (midQ)    at ($(q1)!0.5!(qp)$);
\coordinate (midA)    at ($(a1)!0.5!(ap)$);
\coordinate (midB)    at ($(b1)!0.5!(bm)$);
\coordinate (midV)    at ($(v1)!0.5!(vm)$);
\coordinate (midYbar) at ($(ybar1)!0.5!(ybarm)$);

\foreach \m in {midF,midY,midQ,midA,midB,midV,midYbar}{
  \node[font=\large] at (\m) {$\vdots$};
}

\begin{scope}[on background layer]
\node[predarrow] (predL) at (-2.2,0) {};
\node[predarrow] (predR) at (13.3,0) {};
\end{scope}

\draw[densely dashed, draw=black!60]
  ($(predL.west)+(0,2.2)$) -- ($(predL.west)+(0,-2.2)$);
\node[above] at ($(predL.west)+(0,2.3)$) {$k-1$};

\draw[densely dashed, draw=black!60]
  ($(predR.west)+(0,2.2)$) -- ($(predR.west)+(0,-2.2)$);
\node[above] at ($(predR.west)+(0,2.3)$) {$k$};

\node[font=\tiny] at (predL.center) {prediction};
\node[font=\tiny] at (predR.center) {prediction};

  \draw[very thick,draw=predline]
    (predL.north) -- (predL.north |- f1.west) -- (f1.west);                             

\draw[very thick,draw=predline]
    (predL.south)
    -- (predL.south |- fp.west)
    -- (fp.west);

\draw[very thick,draw=predline]
  (y1.south) -- ++(0,-0.3) coordinate (yL0)
  -- (predR.west |- yL0);

\draw[very thick,draw=predline]
  (yp.north) -- ++(0,0.3) coordinate (yL1)
  -- (predR.west |- yL1);

\draw[very thick,draw=predline]
  (ybar1.south) -- ++(0,-0.5) coordinate (yR0)
  -- (predR.west |- yR0);

\draw[very thick,draw=predline]
  (ybarm.north) -- ++(0,0.5) coordinate (yR1)
  -- (predR.west |- yR1);

\end{tikzpicture}

%% file: main_combined.bbl
\begin{thebibliography}{10}
\providecommand{\url}[1]{#1}
\csname url@samestyle\endcsname
\providecommand{\newblock}{\relax}
\providecommand{\bibinfo}[2]{#2}
\providecommand{\BIBentrySTDinterwordspacing}{\spaceskip=0pt\relax}
\providecommand{\BIBentryALTinterwordstretchfactor}{4}
\providecommand{\BIBentryALTinterwordspacing}{\spaceskip=\fontdimen2\font plus
\BIBentryALTinterwordstretchfactor\fontdimen3\font minus
  \fontdimen4\font\relax}
\providecommand{\BIBforeignlanguage}[2]{{%
\expandafter\ifx\csname l@#1\endcsname\relax
\typeout{** WARNING: IEEEtran.bst: No hyphenation pattern has been}%
\typeout{** loaded for the language `#1'. Using the pattern for}%
\typeout{** the default language instead.}%
\else
\language=\csname l@#1\endcsname
\fi
#2}}
\providecommand{\BIBdecl}{\relax}
\BIBdecl

\bibitem{bai2025beliefpropagationbasedtargethandover}
L.~Bai, Y.~Ge, and H.~Wymeersch, ``Belief propagation-based target handover in
  distributed integrated sensing and communication,'' in \emph{Proc. IEEE
  Global Commun. Conf. (GLOBECOM)}, 2025, pp. 752--757.

\bibitem{ISAC}
N.~Gonz{\'a}lez-Prelcic \emph{et~al.}, ``The integrated sensing and
  communication revolution for {6G}: Vision, techniques, and applications,''
  \emph{Proc. IEEE}, vol. 112, no.~7, pp. 676--723, Jul. 2024.

\bibitem{meng2024cooperativesensing}
K.~Meng and C.~Masouros, ``Cooperative sensing and communication for {ISAC}
  networks: Performance analysis and optimization,'' in \emph{Proc. IEEE 25th
  Int. Workshop Signal Process. Adv. Wireless Commun. (SPAWC)}, 2024, pp.
  446--450.

\bibitem{meng2025cooperativeisacscaling}
K.~Meng, C.~Masouros, A.~P. Petropulu, and L.~Hanzo, ``Cooperative {ISAC}
  networks: Performance analysis, scaling laws, and optimization,'' \emph{IEEE
  Trans. Wireless Commun.}, vol.~24, no.~2, pp. 877--892, Feb. 2025.

\bibitem{strinati2024distributedintelligentintegratedsensing}
E.~C. Strinati \emph{et~al.}, ``Toward distributed and intelligent integrated
  sensing and communications for {6G} networks,'' \emph{IEEE Wireless Commun.},
  vol.~32, no.~1, pp. 60--67, Feb. 2025.

\bibitem{strinati2024disacapproach}
E.~C. Strinati \emph{et~al.}, ``Distributed intelligent integrated sensing and
  communications: The {6G}-{DISAC} approach,'' in \emph{Proc. Joint Eur. Conf.
  Netw. Commun. \& 6G Summit (EuCNC/6G Summit)}, 2024, pp. 392--397.

\bibitem{granstrom2017extendedobjecttrackingintroduction}
K.~Granstr{\"o}m, M.~Baum, and S.~Reuter, ``Extended object tracking:
  Introduction, overview, and applications,'' \emph{J. Adv. Inf. Fusion},
  vol.~12, no.~2, pp. 139--174, Dec. 2017.

\bibitem{granstrom2012gmphd}
K.~Granstr{\"o}m, C.~Lundquist, and U.~Orguner, ``Extended target tracking
  using a {Gaussian}-mixture {PHD} filter,'' \emph{IEEE Trans. Aerosp.
  Electron. Syst.}, vol.~48, no.~4, pp. 3268--3286, Oct. 2012.

\bibitem{lundquist2013cphd}
C.~Lundquist, K.~Granstr{\"o}m, and U.~Orguner, ``An extended target {CPHD}
  filter and a {Gamma} {Gaussian} inverse {Wishart} implementation,''
  \emph{IEEE J. Sel. Topics Signal Process.}, vol.~7, no.~3, pp. 472--483, Jun.
  2013.

\bibitem{granstrom2012phd}
K.~Granstr{\"o}m and U.~Orguner, ``A {PHD} filter for tracking multiple
  extended targets using random matrices,'' \emph{IEEE Trans. Signal Process.},
  vol.~60, no.~11, pp. 5657--5671, Nov. 2012.

\bibitem{beard2016labeledeot}
M.~Beard, S.~Reuter, K.~Granstr{\"o}m, B.-T. Vo, B.-N. Vo, and A.~Scheel,
  ``Multiple extended target tracking with labeled random finite sets,''
  \emph{IEEE Trans. Signal Process.}, vol.~64, no.~7, pp. 1638--1653, Apr.
  2016.

\bibitem{granstrom2020pmbm}
K.~Granstr{\"o}m, M.~Fatemi, and L.~Svensson, ``{Poisson} multi-{Bernoulli}
  mixture conjugate prior for multiple extended target filtering,'' \emph{IEEE
  Trans. Aerosp. Electron. Syst.}, vol.~56, no.~1, pp. 208--225, Feb. 2020.

\bibitem{xia2022pmbapprox}
Y.~Xia, K.~Granstr{\"o}m, L.~Svensson, M.~Fatemi, {\'A}.~F.
  Garc{\'i}a-Fern{\'a}ndez, and J.~L. Williams, ``{Poisson} multi-{Bernoulli}
  approximations for multiple extended object filtering,'' \emph{IEEE Trans.
  Aerosp. Electron. Syst.}, vol.~58, no.~2, pp. 890--906, Apr. 2022.

\bibitem{sampling_method}
K.~Granstr{\"o}m, L.~Svensson, S.~Reuter, Y.~Xia, and M.~Fatemi,
  ``Likelihood-based data association for extended object tracking using
  sampling methods,'' \emph{IEEE Trans. Intell. Veh.}, vol.~3, no.~1, pp.
  30--45, Mar. 2018.

\bibitem{xia2023trajectorybp}
Y.~Xia, {\'A}.~F. Garc{\'i}a-Fern{\'a}ndez, F.~Meyer, J.~L. Williams,
  K.~Granstr{\"o}m, and L.~Svensson, ``Trajectory {PMB} filters for extended
  object tracking using belief propagation,'' \emph{IEEE Trans. Aerosp.
  Electron. Syst.}, vol.~59, no.~6, pp. 9312--9331, Dec. 2023.

\bibitem{meyer2020scalable}
F.~Meyer and M.~Z. Win, ``Scalable data association for extended object
  tracking,'' \emph{IEEE Trans. Signal Inf. Process. Netw.}, vol.~6, pp.
  491--507, 2020.

\bibitem{meyer2021scalable}
F.~Meyer and J.~L. Williams, ``Scalable detection and tracking of geometric
  extended objects,'' \emph{IEEE Trans. Signal Process.}, vol.~69, pp.
  6283--6298, 2021.

\bibitem{fusion}
G.~Koliander, Y.~El-Laham, P.~M. Djuri{\'c}, and F.~Hlawatsch, ``Fusion of
  probability density functions,'' \emph{Proc. IEEE}, vol. 110, no.~4, pp.
  404--453, Apr. 2022.

\bibitem{moratuwage2022multi}
D.~Moratuwage, B.-N. Vo, B.-T. Vo, and C.~Shim, ``Multi-scan multi-sensor
  multi-object state estimation,'' \emph{IEEE Trans. Signal Process.}, vol.~70,
  pp. 5429--5442, 2022.

\bibitem{nannuru2016multisensor}
S.~Nannuru, S.~Blouin, M.~Coates, and M.~Rabbat, ``Multisensor {CPHD} filter,''
  \emph{IEEE Trans. Aerosp. Electron. Syst.}, vol.~52, no.~4, pp. 1834--1854,
  Aug. 2016.

\bibitem{vo2019multisensorglmb}
B.-N. Vo, B.-T. Vo, and M.~Beard, ``Multi-sensor multi-object tracking with the
  generalized labeled multi-{Bernoulli} filter,'' \emph{IEEE Trans. Signal
  Process.}, vol.~67, no.~23, pp. 5952--5967, Dec. 2019.

\bibitem{frohle2019multisensorpmb}
M.~Fr{\"o}hle, C.~Lindberg, K.~Granstr{\"o}m, and H.~Wymeersch, ``Multisensor
  {Poisson} multi-{Bernoulli} filter for joint target--sensor state tracking,''
  \emph{IEEE Trans. Intell. Veh.}, vol.~4, no.~4, pp. 609--621, Dec. 2019.

\bibitem{BraGagSolRicGabLepNicWilBraWin:JSP2022}
M.~Brambilla \emph{et~al.}, ``Cooperative localization and multitarget tracking
  in agent networks with the sum-product algorithm,'' \emph{IEEE Open J. Signal
  Process.}, vol.~3, pp. 169--195, Mar. 2022.

\bibitem{uney2013distributed}
M.~Uney, D.~E. Clark, and S.~J. Julier, ``Distributed fusion of {PHD} filters
  via exponential mixture densities,'' \emph{IEEE J. Sel. Topics Signal
  Process.}, vol.~7, no.~3, pp. 521--531, Jun. 2013.

\bibitem{battistelli2013consensus_cphd}
G.~Battistelli, L.~Chisci, C.~Fantacci, A.~Farina, and A.~Graziano, ``Consensus
  {CPHD} filter for distributed multitarget tracking,'' \emph{IEEE J. Sel.
  Topics Signal Process.}, vol.~7, no.~3, pp. 508--520, Jun. 2013.

\bibitem{li2021distributedmultiviewcphd}
G.~Li, G.~Battistelli, L.~Chisci, W.~Yi, and L.~Kong, ``Distributed multi-view
  multi-target tracking based on {CPHD} filtering,'' \emph{Signal Process.},
  vol. 188, 2021, {Art. no.}~108210.

\bibitem{fantacci2018distributed}
C.~Fantacci, B.-N. Vo, B.-T. Vo, G.~Battistelli, and L.~Chisci, ``Robust fusion
  for multisensor multiobject tracking,'' \emph{IEEE Signal Process. Lett.},
  vol.~25, no.~5, pp. 640--644, May 2018.

\bibitem{saucan2018distributed}
A.-A. Saucan and P.~K. Varshney, ``Distributed cross-entropy $\delta$-{GLMB}
  filter for multi-sensor multi-target tracking,'' in \emph{Proc. 21st Int.
  Conf. Inf. Fusion (FUSION)}, 2018, pp. 1559--1566.

\bibitem{frohle2020decentralizedpmb}
M.~Fr{\"o}hle, K.~Granstr{\"o}m, and H.~Wymeersch, ``Decentralized {Poisson}
  multi-{Bernoulli} filtering for vehicle tracking,'' \emph{IEEE Access},
  vol.~8, pp. 126\,414--126\,427, 2020.

\bibitem{garcia2025distributed}
{\'A}.~F. Garc{\'i}a-Fern{\'a}ndez and G.~Battistelli, ``Distributed {Poisson}
  multi-{Bernoulli} filtering via generalized covariance intersection,''
  \emph{IEEE Trans. Signal Process.}, vol.~74, pp. 246--257, 2026.

\bibitem{van2021distributed}
H.~V. Nguyen, H.~Rezatofighi, B.-N. Vo, and D.~C. Ranasinghe, ``Distributed
  multi-object tracking under limited field of view sensors,'' \emph{IEEE
  Trans. Signal Process.}, vol.~69, pp. 5329--5344, 2021.

\bibitem{chen2025distributed}
F.~Chen, H.~V. Nguyen, A.~S. Leong, S.~Panicker, R.~Baker, and D.~C.
  Ranasinghe, ``Distributed multi-object tracking under limited field of view
  heterogeneous sensors with density clustering,'' \emph{Signal Process.}, vol.
  228, 2025, {Art. no.}~109703.

\bibitem{xiong2024distributedlimitedfov}
C.~Xiong, M.~Hu, H.~Lu, and F.~Zhao, ``Distributed multi-sensor fusion for
  multi-group/extended target tracking with different limited fields of view,''
  \emph{Appl. Sci.}, vol.~14, no.~21, 2024, {Art. no.}~9627.

\bibitem{SharmaTSP2019}
P.~Sharma, A.-A. Saucan, D.~J. Bucci Jr., and P.~K. Varshney, ``Decentralized
  {Gaussian} filters for cooperative self-localization and multi-target
  tracking,'' \emph{IEEE Trans. Signal Process.}, vol.~67, no.~22, pp.
  5896--5911, Nov. 2019.

\bibitem{amosa2023multicamera}
T.~I. Amosa \emph{et~al.}, ``Multi-camera multi-object tracking: A review of
  current trends and future advances,'' \emph{Neurocomputing}, vol. 552, 2023,
  {Art. no.}~126558.

\bibitem{wang2022camerahandoff}
P.-T. Wang, J.-S. Sheu, and J.-H. Lai, ``Camera handoff for multicamera
  multiobject tracking,'' \emph{Sens. Mater.}, vol.~34, no.~2, pp. 563--574,
  Feb. 2022.

\bibitem{tesfaye2019multitarget}
Y.~T. Tesfaye, E.~Zemene, A.~Prati, M.~Pelillo, and M.~Shah, ``Multi-target
  tracking in multiple non-overlapping cameras using fast-constrained dominant
  sets,'' \emph{Int. J. Comput. Vis.}, vol. 127, no.~9, pp. 1303--1320, Sep.
  2019.

\bibitem{ge2024targethandoverdistributedintegrated}
Y.~Ge \emph{et~al.}, ``Target handover in distributed integrated sensing and
  communication,'' in \emph{Proc. IEEE Int. Conf. Commun. (ICC)}, 2025, pp.
  3960--3965.

\bibitem{kim2019optimal}
J.~Kim, D.-H. Cho, W.-C. Lee, S.-S. Park, and H.-L. Choi, ``Optimal target
  assignment with seamless handovers for networked radars,'' \emph{Sensors},
  vol.~19, no.~20, Oct. 2019, {Art. no.}~4555.

\bibitem{fan_handover}
C.~Zhao, Y.~Feng, H.~Luo, F.~Gao, F.~Liu, and S.~Jin, ``Networked {ISAC}-based
  {UAV} tracking and handover toward low-altitude economy,'' \emph{IEEE Trans.
  Wireless Commun.}, vol.~24, no.~9, pp. 7670--7685, Sep. 2025.

\bibitem{ribeiro2025mobilitymanagementintegratedsensing}
\BIBentryALTinterwordspacing
Y.~S. Ribeiro, B.~Makki, A.~L.~F. de~Almeida, {Fazal-E-Asim}, and G.~Fodor,
  ``Mobility management in integrated sensing and communications networks,''
  2025, {arXiv:2501.08159}. [Online]. Available:
  \url{https://arxiv.org/abs/2501.08159}
\BIBentrySTDinterwordspacing

\bibitem{liesegang2026scalableisac}
S.~Liesegang, S.~Buzzi, and C.~D'Andrea, ``Scalable integrated sensing and
  communications for multi-target detection and tracking in cell-free massive
  {MIMO}: A unified framework,'' \emph{IEEE Trans. Commun.}, vol.~74, pp.
  2777--2793, 2026.

\bibitem{venugopal2017channel}
K.~Venugopal, A.~Alkhateeb, N.~Gonz{\'a}lez-Prelcic, and R.~W. Heath Jr.,
  ``Channel estimation for hybrid architecture-based wideband millimeter wave
  systems,'' \emph{IEEE J. Sel. Areas Commun.}, vol.~35, no.~9, pp. 1996--2009,
  Sep. 2017.

\bibitem{jiang2021beamspace}
\BIBentryALTinterwordspacing
F.~Jiang, F.~Wen, Y.~Ge, M.~Zhu, H.~Wymeersch, and F.~Tufvesson, ``Beamspace
  multidimensional {ESPRIT} approaches for simultaneous localization and
  communications,'' 2021, {arXiv:2111.07450}. [Online]. Available:
  \url{https://arxiv.org/abs/2111.07450}
\BIBentrySTDinterwordspacing

\bibitem{ester1996density}
M.~Ester, H.-P. Kriegel, J.~Sander, and X.~Xu, ``A density-based algorithm for
  discovering clusters in large spatial databases with noise,'' in \emph{Proc.
  2nd Int. Conf. Knowl. Discovery Data Mining (KDD)}.\hskip 1em plus 0.5em
  minus 0.4em\relax AAAI Press, 1996, pp. 226--231.

\bibitem{richards2014fundamentals}
M.~A. Richards, \emph{Fundamentals of Radar Signal Processing}, 2nd~ed.\hskip
  1em plus 0.5em minus 0.4em\relax New York, NY, USA: McGraw-Hill Education,
  2014.

\bibitem{meyer2018messagepassing}
F.~Meyer \emph{et~al.}, ``Message passing algorithms for scalable multitarget
  tracking,'' \emph{Proc. IEEE}, vol. 106, no.~2, pp. 221--259, Feb. 2018.

\end{thebibliography}

\begin{thebibliography}{1}
\providecommand{\url}[1]{#1}
\csname url@samestyle\endcsname
\providecommand{\newblock}{\relax}
\providecommand{\bibinfo}[2]{#2}
\providecommand{\BIBentrySTDinterwordspacing}{\spaceskip=0pt\relax}
\providecommand{\BIBentryALTinterwordstretchfactor}{4}
\providecommand{\BIBentryALTinterwordspacing}{\spaceskip=\fontdimen2\font plus
\BIBentryALTinterwordstretchfactor\fontdimen3\font minus
  \fontdimen4\font\relax}
\providecommand{\BIBforeignlanguage}[2]{{%
\expandafter\ifx\csname l@#1\endcsname\relax
\typeout{** WARNING: IEEEtran.bst: No hyphenation pattern has been}%
\typeout{** loaded for the language `#1'. Using the pattern for}%
\typeout{** the default language instead.}%
\else
\language=\csname l@#1\endcsname
\fi
#2}}
\providecommand{\BIBdecl}{\relax}
\BIBdecl

\bibitem{meyer2018messagepassing}
F.~Meyer \emph{et~al.}, ``Message passing algorithms for scalable multitarget
  tracking,'' \emph{Proc. IEEE}, vol. 106, no.~2, pp. 221--259, Feb. 2018.

\bibitem{meyer2018messagepassingSupplement}
\BIBentryALTinterwordspacing
F.~Meyer \emph{et~al.}, ``Message passing algorithms for scalable multitarget
  tracking: Supplementary material,'' Jul.~2, 2019, {Errata and supplementary
  material}. [Online]. Available:
  \url{https://winslab.lids.mit.edu/wp-content/uploads/2019/07/ProcIEEE_MTT_Suppl_Mat.pdf}
\BIBentrySTDinterwordspacing

\bibitem{meyer2021scalable}
F.~Meyer and J.~L. Williams, ``Scalable detection and tracking of geometric
  extended objects,'' \emph{IEEE Trans. Signal Process.}, vol.~69, pp.
  6283--6298, 2021.

\bibitem{Yuxuanxia_standardmodel}
{\'A}.~F. Garc{\'i}a-Fern{\'a}ndez, J.~L. Williams, L.~Svensson, and Y.~Xia,
  ``A {Poisson} multi-{Bernoulli} mixture filter for coexisting point and
  extended targets,'' \emph{IEEE Trans. Signal Process.}, vol.~69, pp.
  2600--2610, 2021.

\end{thebibliography}
